\documentclass[%
reprint,
superscriptaddress,
bibnotes,
amsmath,amssymb,
aps,
prx,
floatfix,
]{revtex4-2}
\usepackage{graphicx}
\usepackage{dcolumn}
\usepackage{bm}
\usepackage[breaklinks=true]{hyperref}
\usepackage{natbib}
\usepackage{mathrsfs,mathtools,color,wasysym,dsfont,here}
\usepackage{amsthm}
\usepackage{physics}
\usepackage[all]{xy}
\usepackage{amsfonts}
\usepackage{array,booktabs}
\usepackage{comment}
\usepackage{cases}

\newcommand{\ad}{\operatorname{ad}}
\renewcommand{\Tr}{\operatorname{Tr}}
\newcommand{\Ker}{\operatorname{Ker}}
\newcommand{\opn}{|\!|\!|}

\newcommand{\cst}{\mathcal{C}^\mathrm{Sch} }
\newcommand{\aast}{\mathcal{A}^\mathrm{Sch} }
\newcommand{\cint}{\mathcal{C}^\mathrm{Int} }
\newcommand{\aint}{\mathcal{A}^\mathrm{Int} }

\theoremstyle{plain}
\newtheorem{thm}{Theorem}
\newtheorem{lem}[thm]{Lemma}

\theoremstyle{definition}
\newtheorem{dfn}[thm]{Definition}
\newtheorem{con}[thm]{Condition}

\theoremstyle{remark}

\newtheorem{ex}{Example}

\begin{document}

\title{Theory of Steady States for Lindblad Equations beyond Time-Independence: Classification, Uniqueness and Symmetry}

\author{Hironobu Yoshida}
\email{hironobu-yoshida57@g.ecc.u-tokyo.ac.jp}

\affiliation{Department of Physics, Graduate School of Science, The University of Tokyo, 7-3-1, Hongo, Bunkyo-ku, Tokyo, 113-0033, Japan}

\affiliation{Nonequilibrium Quantum Statistical Mechanics RIKEN Hakubi Research Team,
RIKEN Pioneering Research Institute (PRI), Wako, Saitama 351-0198, Japan}

\author{Ryusuke Hamazaki}
\affiliation{Nonequilibrium Quantum Statistical Mechanics RIKEN Hakubi Research Team,
RIKEN Pioneering Research Institute (PRI), Wako, Saitama 351-0198, Japan}
\affiliation{RIKEN Center for Interdisciplinary Theoretical and Mathematical Sciences (iTHEMS), RIKEN, Wako 351-0198, Japan}
\begin{abstract}
We present a rigorous and comprehensive classification of the asymptotic behavior of time-dependent Gorini-Kossakowski-Sudarshan-Lindblad (GKSL) equations under the assumption of Hermitian jump operators. Our results apply to a broad class of GKSL equations whose time dependence is assumed to be recurrent, including time-independent, periodic, quasiperiodic, and certain classes of random time dependence.
Our main contributions are twofold: first, we establish a criterion for the uniqueness of steady states. The criterion is formulated in terms of the algebra generated by the GKSL generators and provides a necessary and sufficient condition when the generators are analytic functions of time. 
We demonstrate the utility of our criterion through prototypical examples, including quantum many-body spin chains.
Second, we extend the concept of strong symmetry for time-dependent GKSL equations by introducing two distinct forms,  strong symmetry in the Schr\"odinger picture and that in the interaction picture, and completely classify the asymptotic dynamics with them.  
More concretely, we rigorously uncover that the strong symmetry in the interaction picture is responsible for non-trivial \textit{time-dependent} steady states, such as coherent oscillations, whereas that in the Schr\"odinger picture controls the existence of time-independent steady states. 
This classification not only encompasses established mechanisms underlying non-trivial oscillatory steady states, such as strong dynamical symmetry and Floquet dynamical symmetry, but also reveals symmetry-predicted, time-dependent asymptotic dynamics in a novel class of open quantum systems. Our framework thus provides a rigorous foundation for controlling dissipative quantum systems in a time-dependent manner.
\end{abstract}

\maketitle

\section{Introduction}

The Gorini–Kossakowski–Sudarshan–Lindblad~(GKSL) equation~\cite{gorini_CompletelyPositiveDynamical_1976,lindblad_GeneratorsQuantumDynamical_1976} is one of the most fundamental frameworks for modeling open quantum systems. It is microscopically derived under suitable approximations~\cite{breuer_TheoryOpenQuantum_2007,rivas_OpenQuantumSystems_2012,trushechkin_UnifiedGoriniKossakowskiLindbladSudarshanQuantum_2021}, which is physically justified in several cases~\cite{mori_FloquetStatesOpen_2023,shiraishi_QuantumMasterEquation_2025}.
Widely describing open-system dynamics considered to be Markovian, the GKSL equation predicts various interesting phenomena, such as dissipative time crystals~\cite{gong_DiscreteTimeCrystallineOrder_2018,iemini_BoundaryTimeCrystals_2018,kessler_ObservationDissipativeTime_2021a,kongkhambut_ObservationContinuousTime_2022} and quantum synchronization~\cite{tindall_QuantumSynchronisationEnabled_2020,schmolke_NoiseInducedQuantumSynchronization_2022,buca_AlgebraicTheoryQuantum_2022}. In addition to physical implications, its generality and tractability have led researchers to establish a rich body of mathematical results concerning its structure and properties~\cite{wolf_AssessingNonMarkovianQuantum_2008,rivas_OpenQuantumSystems_2012}.

One of the most crucial questions regarding the GKSL equation is whether the steady state is unique or not, which is not only of fundamental interest but also closely related to various applications such as dissipative state preparation~\cite{kraus_PreparationEntangledStates_2008,verstraete_QuantumComputationQuantumstate_2009,harrington_EngineeredDissipationQuantum_2022}. Indeed, if the steady state is unique, every initial state relaxes to it; otherwise, the initial state must be carefully prepared to obtain desired states. Even in the latter case, we may be able to engineer some exotic dynamics such as entangled steady states~\cite{li_HilbertSpaceFragmentation_2023} and constrained dynamics~\cite{stannigel_ConstrainedDynamicsZeno_2014,maity_KineticallyConstrainedModels_2024}.

Owing to their importance, rigorous conditions for the uniqueness of steady states of the time-independent GKSL equation have been studied extensively over a long history~\cite{spohn_AlgebraicConditionApproach_1977,frigerio_QuantumDynamicalSemigroups_1977,frigerio_StationaryStatesQuantum_1978,spohn_KineticEquationsHamiltonian_1980,fagnola_SubharmonicProjectionsQuantum_2002,baumgartner_AnalysisQuantumSemigroups_2008,baumgartner_AnalysisQuantumSemigroups_2008a,nigro_UniquenessSteadystateSolution_2019,amato_AsymptoticsQuantumChannels_2023,yoshida_UniquenessSteadyStates_2024,zhang_CriteriaDaviesIrreducibility_2024,amato_NumberSteadyStates_2024,seltmann_UniqueNESSGraphTheory_2026,hamazaki_IntroductionMonitoredQuantum_2026,fagnola_IrreducibilityQuantumMarkov_2025} beginning with Spohn’s seminal work. These conditions, represented in light of the algebraic structure of the GKSL generators,  have been revisited in recent years and shown to be connected to a distinctive symmetry in open quantum systems called strong symmetry~\cite{buca_NoteSymmetryReductions_2012,albert_SymmetriesConservedQuantities_2014}. It is well known that the presence of strong symmetries leads to degeneracy in the time-independent steady state. 

As another line of research, the so-called strong dynamical symmetries are found to prevent systems from decaying to trivial time-independent steady states
even in the presence of dissipation. In contrast with the strong symmetry, strong dynamical symmetries protect coherent and oscillatory evolution at long times~\cite{buca_NonstationaryCoherentQuantum_2019,piccitto_SymmetriesConservedQuantities_2021,taheri_AllopticalDissipativeDiscrete_2022,hajdusek_SeedingCrystallizationTime_2022,buca_AlgebraicTheoryQuantum_2022} in time-independent GKSL dynamics. Therefore, the study of steady-state structure has evolved beyond a purely theoretical concept and has become a rich source of intriguing physics.

Although substantial progress has been made in the case of time-independent GKSL equations, very little is known about the mathematical properties of \textit{time-dependent} GKSL generators, despite their growing interest. Time-dependent driving leads to richer physical structures,
an important class of which includes time-periodic generators~\cite{prosen_NonequilibriumPhaseTransition_2011,vorberg_GeneralizedBoseEinsteinCondensation_2013,haddadfarshi_CompletelyPositiveApproximate_2015,dai_FloquetTheoremOpen_2016,fitzpatrick_ObservationDissipativePhase_2017,hartmann_AsymptoticFloquetStates_2017,ikeda_GeneralDescriptionNonequilibrium_2020,schnell_ThereFloquetLindbladian_2020,fedorov_PhotonTransportBoseHubbard_2021,schnell_HighfrequencyExpansionsTimeperiodic_2021,mizuta_BreakdownMarkovianityInteractions_2021,ikeda_NonequilibriumSteadyStates_2021,mori_FloquetStatesOpen_2023,chen_PeriodicallyDrivenOpen_2024}. In this case, various states emerge, such as steady states with the same period as the periodic drive or states with periods that are integer multiples of the drive~\cite{gong_DiscreteTimeCrystallineOrder_2018,chinzei_TimeCrystalsProtected_2020,kessler_ObservationDissipativeTime_2021a}. Moreover, steady states can be analyzed using high-frequency expansions, which is useful for dissipative state preparation~\cite{ikeda_NonequilibriumSteadyStates_2021,schnell_HighfrequencyExpansionsTimeperiodic_2021,mori_FloquetStatesOpen_2023}. 
Additionally, quasiperiodically driven systems subjected to external multi-frequency drives~\cite{ho_SemiclassicalManymodeFloquet_1983,casati_AndersonTransitionOneDimensional_1989,martin_TopologicalFrequencyConversion_2017,novicenko_FloquetAnalysisQuantum_2017,cubero_AsymptoticTheoryQuasiperiodically_2018,zhao_FloquetTimeSpirals_2019,nathan_QuasiperiodicFloquetThoulessEnergy_2021,malz_TopologicalTwodimensionalFloquet_2021,beatrez_ObservationCriticalPrethermal_2023,park_ControllableFloquetEdge_2023,roy_SingleMultifrequencyDriving_2024,he_ExperimentalRealizationDiscrete_2025}, Fibonacci drives~\cite{maity_FibonacciSteadyStates_2019,lapierre_FineStructureHeating_2020,pilatowsky-cameo_CompleteHilbertSpaceErgodicity_2023}, and Thue-Morse drives~\cite{mori_RigorousBoundsHeating_2021,zhao_RandomMultipolarDriving_2021,das_PeriodicallyQuasiperiodicallyDriven_2023,yan_PrethermalizationAperiodicallyKicked_2024,tiwari_PeriodicallyAperiodicallyThueMorse_2025,pilatowsky-cameo_CriticallySlowHilbertSpace_2025} have drawn recent theoretical interest. While time-quasiperiodic unitary dynamics have been intensively studied in recent years as they can exhibit exotic dynamics, the properties of time-quasiperiodic GKSL equations remain largely undeveloped, even though dissipation is unavoidable in real systems.

The uniqueness of time-dependent GKSL equations has been addressed in a few studies~\cite{menczel_LimitCyclesPeriodically_2019,meglio_AsymptoticRelaxationQuantum_2024}, which considered sufficient conditions for uniqueness based solely on the jump operators.
However, many fundamental issues remain unresolved.
First, it is crucial to understand how the presence of a Hamiltonian term affects uniqueness, as is well known in the time-independent case~\cite{frigerio_QuantumDynamicalSemigroups_1977,wolf_QuantumChannelsOperations_2012,yoshida_UniquenessSteadyStates_2024,zhang_CriteriaDaviesIrreducibility_2024,hamazaki_IntroductionMonitoredQuantum_2026,fagnola_IrreducibilityQuantumMarkov_2025}.
Moreover, it remains unclear whether necessary conditions for uniqueness exist.
Most importantly, a general formulation of strong symmetries in time-dependent systems, and their relation to the uniqueness and classification of steady states, have never been discussed.

In this work, addressing the above issues with rigorous analysis, we present a novel framework for classifying the steady states in finite-dimensional GKSL equations beyond time-independence. Our assumption is that the generators are recurrent in time and that the jump operators are Hermitian, which includes many important situations (e.g., time-periodic and quasiperiodic GKSL equations with dephasing and depolarizing noise). We first provide an algebraic criterion that provides a necessary and sufficient condition for the uniqueness, using both the Hamiltonian and jump operators. We demonstrate the practical advantage of our criterion through concrete examples, such as a dissipative quantum many-body spin chain. 
We also uncover that there are two distinct forms of time-dependent versions of strong symmetry, i.e., that in the Schr\"odinger picture and interaction picture. We show that the former and the latter symmetries are responsible for time-independent and time-dependent multiple steady states, respectively.
This fact leads to the novel classification of asymptotic behavior of the GKSL dynamics through the two types of strong symmetries.
Our classification encompasses established mechanisms underlying non-trivial oscillatory steady states, such as strong dynamical symmetry.
Furthermore, it predicts a new type of symmetry-predicted, time-dependent asymptotic dynamics through time-quasiperiodic driving, which cannot be understood from the previously known mechanisms.

\begin{figure*}
 \centering
 \includegraphics[width=\linewidth]{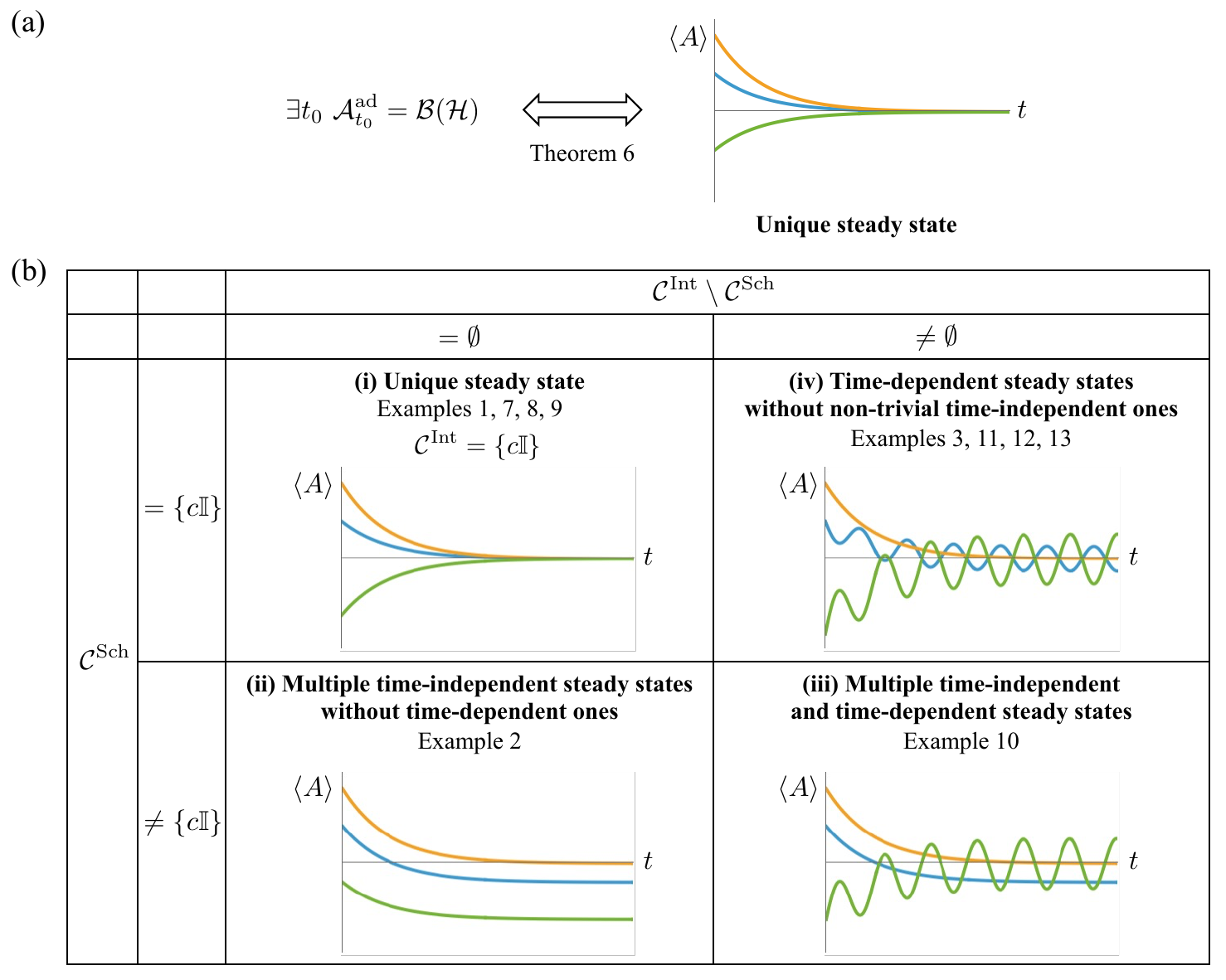}
  \caption{Schematic illustration of our main results. We present a novel framework to analyze the steady-state structure of open quantum systems governed by time-dependent Gorini-Kossakowski-Sudarshan-Lindblad (GKSL) equations with Hermitian jump operators with rigorous analysis. (a) If the algebra $\mathcal{A}^\mathrm{ad}_{t_0}$ defined in \eqref{eq:aad} coincides with the full operator algebra $\mathcal{B}(\mathcal{H})$ for some $t_0$, then every state relaxes to the completely mixed state $\mathbb{I}/d$ for time-recurrent GKSL generators   
  (Theorem~\ref{thm:main2}). Conversely, if $\mathcal{A}^\mathrm{ad}_{t_0}\neq \mathcal{B}(\mathcal{H})$ for all $t_0$, we prove that the steady state is not unique under the assumption that the GKSL generators are analytic at all times.
  (b) For time-dependent GKSL equations, we define strong symmetries characterized by two symmetry algebras: 
  $\cst$ and $\cint$ defined in \eqref{eq:def_cst} and \eqref{eq:def_cint}, respectively. Since $\cst \subseteq \cint$ holds, the following four possibilities arise: 
  $\cst = \{ c\mathbb{I} \mid c \in \mathbb{C} \}$ or $\cst \neq \{ c\mathbb{I} \mid c \in \mathbb{C} \}$, and 
  $\cint \setminus \cst = \emptyset$ or $\cint \setminus \cst \neq \emptyset$. The table summarizes the possible steady states for each case. For example, (iv) corresponds to the existence of time-dependent steady states without non-trivial time-independent ones and represents a class unique to time-dependent GKSL equations. Note that the completely mixed state $\mathbb{I}/d$ is always a trivial steady state. See Definition~\ref{def:steady_state} for the precise meaning of the classification and Sec.~\ref{sec:summary_classification} for an explanation.
  }
\label{fig:schematic}
\end{figure*}

\subsection{Summary of results}
\label{sec:summary}
Here, we summarize the main contributions of this study (see also Fig.~\ref{fig:schematic}).

\subsubsection{Algebraic necessary and sufficient condition for the uniqueness of steady states}
\label{sec:summary_algebraic}
We present a general rigorous criterion for the uniqueness of steady states in GKSL equations with time-dependent generators, under the assumption of Hermitian jump operators (Theorem~\ref{thm:main2} in Sec.~\ref{sec:useful_criterion}).  
The GKSL equation is given by
\begin{align}
  &\partial_t \rho_t
  =\mathcal{L}_t (\rho_t)\nonumber\\
  &=-i[H_t, \rho_t]+\sum_{m=1}^M \left(L_{m,t} \rho_t L_{m,t}-\frac{1}{2}\left\{(L_{m,t})^2, \rho_t\right\}\right),
\end{align}
where $t$ denotes time, $H_t$ is the Hamiltonian, and $L_{m,t}$ are the jump operators. The superoperator $\mathcal{L}_t$ is called the Liouvillian, which we mainly assume to be time-recurrent, including time-independent, periodic, quasiperiodic, and certain classes of random time dependences (see Condition~\ref{con:recurrent} for the definition and Examples~\ref{ex:multi_frequency_drives}~and~\ref{ex:recurrent_drives} for examples).

We first show the condition for the uniqueness by introducing
 the following algebra $\mathcal{A}^\mathrm{ad}_{t}$:
\begin{align}\label{eq:ad_intro}
\mathcal{A}_{t}^{\mathrm{ad}}=\langle\,\mathbb{I},\{\ad^n_{t}(L_{m,t})\}_{m=1,\ldots,M;n=0,1,\ldots}\rangle,
\end{align}
where $\ad_{t}$ is the adjoint operation defined as
\begin{equation}    
    \ad_{t}(A):=i[H_{t},A]+\partial_t A.
\end{equation}
If $\mathcal{A}^\mathrm{ad}_{t}$ coincides with the full operator algebra $\mathcal{B}(\mathcal{H})$ for some $t=t_0$, then every state relaxes to the completely mixed state $\mathbb{I}/d$.

Note that this is a non-trivial generalization of the time-independent case, where the uniqueness of steady states is characterized by~Eq.~\eqref{eq:algebra_indep}; indeed, Eq.~\eqref{eq:ad_intro} includes the time-derivative factor absent in \eqref{eq:algebra_indep}. When the GKSL generators are analytic at all times, the criterion provides a necessary and sufficient condition for the uniqueness of steady states (see Fig.~\ref{fig:schematic}(a)). 

We illustrate the criterion’s applicability through examples involving two-level systems and a quantum spin chain, for which the uniqueness of the steady states cannot be proven by the existing methods in Refs.~\cite{menczel_LimitCyclesPeriodically_2019,meglio_AsymptoticRelaxationQuantum_2024}, which do not use information from the Hamiltonian part of the dynamics.
    
We remark that, for the periodic GKSL equation, integrating over one period of it leads to a time-independent quantum channel, and the problem of the uniqueness of the steady state reduces to the problem of the uniqueness of the fixed point of the channel~\cite{wolf_QuantumChannelsOperations_2012,burgarth_ErgodicMixingQuantum_2013}. However, such an integration procedure is generally difficult to perform, and thus, methods for determining uniqueness directly from the generator of the GKSL equation are practically useful. Furthermore, the integration procedure is not applicable to the non-periodic GKSL equation, which highlights the fundamental advantage of our method.

\subsubsection{Generalization of the strong symmetry for time-dependent GKSL equations}
\label{sec:summary_generalization}
For time-independent GKSL equations, the strong symmetry is defined as a symmetry that commutes with both the Hamiltonian and all the jump operators that describe dissipation due to coupling to the environment. Mathematically, strong symmetry is characterized by the commutant algebra. Non-trivial strong symmetry gives rise to conserved quantities, which lead to degeneracy of the steady states. Conversely, for time-independent GKSL equations, when the jump operators are Hermitian, the steady state is unique in the absence of non-trivial strong symmetry (see Sec.~\ref{sec:review}).
    
In Sec.~\ref{sec:insufficiency}, we point out that it is non-trivial to generalize strong symmetry to time-dependent GKSL equations while retaining similar properties. To fill the gap, we identify distinct strong symmetries characterized by two commutant algebras, $\cst$ and $\cint$ defined as
\begin{align}
 \cst &=\{O\in\mathcal{B}(\mathcal{H})\mid [H_t,O]=[L_{m,t},O]=0\ \forall m,t\},\\
  \cint &=\{O\in\mathcal{B}(\mathcal{H})\mid [\tilde{L}_{m,t},O]=0\ \forall m,t\},
\end{align}
where
\begin{align}
    U_t=\mathcal{T}e^{-i\int_0^t H_{t^{\prime}}dt^\prime},\ \tilde{L}_{m,t}=U_t^\dagger L_{m,t} U_t.
\end{align}
We refer to $\cst$ as the strong symmetry in the Schr\"odinger picture, and $\cint$ as the strong symmetry in the interaction picture. They satisfy the inclusion relation $\cst\subseteq \cint$. As we see below, this framework is useful for classifying steady states when the GKSL generator is recurrent in time. 
For example, we can prove that there exist time-dependent steady states if and only if $\cint \setminus \cst \neq \emptyset$ (Theorem~\ref{thm:time_dependent} in Sec.~\ref{sec:existence}).

\subsubsection{Classification of the steady-state structure in terms of the strong symmetry}
\label{sec:summary_classification}
We rigorously elucidate that the steady-state structure of time-recurrent GKSL equations is characterized by two types of strong symmetries $\cst$ and $\cint$ (Theorems~\ref{thm:indep} and \ref{thm:time_dependent} in Sec.~\ref{sec:implication}). Since $\cst \subseteq \cint$ holds, the following possibilities arise: $\cst = \{ c\mathbb{I} \mid c \in \mathbb{C} \}$ or $\cst \neq \{ c\mathbb{I} \mid c \in \mathbb{C} \}$, and $\cint \setminus \cst = \emptyset$ or $\cint \setminus \cst \neq \emptyset$. Figure~\ref{fig:schematic}(b) summarizes the possible steady states for each case in the sense of Definition~\ref{def:steady_state}. We also provide concrete examples for each of the cases.
Note that the assumption of time-recurrence is not merely a technical one; in non-recurrent cases, there are counterexamples of our theorem (Example~\ref{ex:exp_decay_dephasing} in Sec.~\ref{sec:necessary_and_sufficient}). 

The top-left panel [class (i)] shows the case where $\cst = \{ c\mathbb{I} \mid c \in \mathbb{C} \}$ and $\cint \setminus \cst = \emptyset$, which is equivalent to $\cint = \{ c\mathbb{I} \mid c \in \mathbb{C} \}$. We prove in Theorem~\ref{thm:main1} that the steady state is unique if and only if this condition is satisfied, which is a non-trivial generalization of the time-independent case. 
    
The bottom-left panel [class (ii)] illustrates the case where $\cst \neq \{ c\mathbb{I} \mid c \in \mathbb{C} \}$ and $\cint \setminus \cst = \emptyset$. In this case, there are multiple time-independent steady states, whereas there are no time-dependent steady states.

The top-right panel [class (iv)] shows the case where $\cst = \{ c\mathbb{I} \mid c \in \mathbb{C} \}$ and $\cint \setminus \cst \neq \emptyset$. In this case, there must be a time-dependent steady state, but there must not be a time-independent steady state other than the maximally mixed state. Notably, this case is unique to time-dependent GKSL equations, because for time-independent generators, coherently oscillating modes are always accompanied by multiple time-independent steady states~\cite{wolf_QuantumChannelsOperations_2012}.
Furthermore, we also provide more non-trivial dynamics belonging to this class, analyzing quasiperiodically driven many-body systems.
    
Finally, the bottom-right panel [class (iii)] illustrates the case where $\cst \neq \{ c\mathbb{I} \mid c \in \mathbb{C} \}$ and $\cint \setminus \cst \neq \emptyset$. In this situation, there exist multiple steady states, both time-independent and time-dependent. Coherent oscillations within time-independent GKSL equations are classified here. It is known that strong dynamical symmetry gives rise to such oscillations~\cite{buca_NonstationaryCoherentQuantum_2019,chinzei_TimeCrystalsProtected_2020,buca_AlgebraicTheoryQuantum_2022}, and one can verify that this symmetry ensures $\cst \neq \{ c\mathbb{I} \mid c \in \mathbb{C} \}$ and $\cint \setminus \cst \neq \emptyset$ (see Sec.~\ref{sec:strong_dynamical_symmetry}). Therefore, strong dynamical symmetry can be reinterpreted within the framework of our classification.

The remainder of the paper is organized as follows.
We begin with defining how to classify the steady-state structure of time-dependent GKSL equations with Hermitian jump operators in Sec.~\ref{sec:classification}. 
After the brief review of the theory of the uniqueness in the case of time-independent GKSL equations in Sec.~\ref{sec:review}, we provide a criterion for the unique steady states for time-recurrent GKSL equations in Sec.~\ref{sec:criteria_for_unique}. First, we prove that the steady state is unique if and only if $\cint=\{c\mathbb{I}|c\in\mathbb{C}\}$ (Theorem~\ref{thm:main1}). Then, we rephrase the condition in the language of the GKSL generators $\{H_t,L_{m,t}\}$ in the Schr\"odinger picture (Theorem~\ref{thm:main2}), whose practical advantage is demonstrated with concrete examples. 
In Sec.~\ref{sec:implication}, we prove two theorems that justify the classification of the steady states given in Sec.~\ref{sec:classification} from the two distinct strong symmetries, with providing concrete applications. Specifically, Theorem~\ref{thm:indep} ensures the absence of multiple time-independent steady states under $\cst=\{c\mathbb{I}|c\in\mathbb{C}\}$ and Theorem~\ref{thm:time_dependent} states that there are time-dependent steady states if and only if $\cint\setminus\cst\neq \emptyset$. Finally, Sec.~\ref{sec:conclusion} presents a conclusion and outlook. All detailed mathematical proofs and calculations can be found in the Appendices.

\section{Classification of steady states for time-dependent Liouvillians with Hermitian jump operators}
\label{sec:classification}
Throughout the manuscript, we focus on open quantum systems in a finite-dimensional Hilbert space $\mathcal{H}$ with $\mathrm{dim}[\mathcal{H}]=d$, whose time evolution is governed by the GKSL equation in the Schr\"odinger picture of the form
\begin{align}
  &\partial_t \rho_t
  =\mathcal{L}_t (\rho_t)\nonumber\\
  &=-i[H_t, \rho_t]+\sum_{m=1}^M \left(L_{m,t} \rho_t L_{m,t}-\frac{1}{2}\left\{(L_{m,t})^2, \rho_t\right\}\right).
  \label{eq:GKSL}
\end{align}
Here, $t$ denotes time, $H_t$ is the Hamiltonian, and $L_{m,t}$ are the jump operators that act on $\mathcal{H}$. The superoperator $\mathcal{L}_t$ is called the Liouvillian. We assume that all $L_{m,t}$ are Hermitian.
Examples described by this equation include a quantum system driven by an external field in the presence of dephasing, as well as a time-dependent quantum system subjected to classical white noise~\cite{caldeira_PathIntegralApproach_1983,lashkari_FastScramblingConjecture_2013,xu_LocalityQuantumFluctuations_2019,ogunnaike_UnifyingEmergentHydrodynamics_2023,moudgalya_SymmetriesGroundStates_2024,chen_StrongtoweakSymmetryBreaking_2025,ziereis_StrongtoWeakSymmetryBreaking_2025}.  In particular, in the latter situation, one can naturally engineer time-dependent jump operators by coupling white noise to a time-dependent external field~\cite{chenu_QuantumSimulationGeneric_2017,schmolke_NoiseInducedQuantumSynchronization_2022}.

In this section, we propose a classification of long-time steady states.
For this purpose, we define steady states  $\rho^*_t$, which can be time-dependent, as those satisfying
\begin{equation}
    \lim_{t\to \infty} (\rho_t-\rho^*_t)=0
    \label{eq:def_steady_state}
\end{equation}
for some initial condition $\rho_0$. In what follows, the notation $\lim_{t\to \infty} \rho_t=\rho^*_t$ is used to mean \eqref{eq:def_steady_state}. Again, $\rho^*_t=\mathbb{I}/d$ is always a time-independent steady state by taking $\rho_0=\mathbb{I}/d$, due to the Hermiticity of $L_m$. Therefore, the asymptotic relaxation of the system can be classified into the following four scenarios (see Fig.~\ref{fig:schematic}):
\begin{dfn} We say that the GKSL equation has
\label{def:steady_state}
  \begin{enumerate} 
    \item[(i)] a unique steady state if
    \begin{equation}
        \lim_{t\to \infty} \rho_t=\mathbb{I}/d
    \end{equation}
    for any initial condition $\rho_0$.
    \item[(ii)] multiple time-independent steady states without time-dependent ones if every initial state relaxes to a time-independent state, and there exists $\rho^*\neq\mathbb{I}/d$ such that
    \begin{equation}
        \lim_{t\to \infty} \rho_t=\rho^*,
    \end{equation}
    for some initial condition.
    \item[(iii)] multiple time-independent and time-dependent steady states if there exist time-independent $\rho^*\neq\mathbb{I}/d$ such that
    \begin{equation}
        \lim_{t\to \infty} \rho_t=\rho^*,
    \end{equation}
    for some initial condition and time-dependent $\rho^*_t$ such that 
    \begin{equation}
        \lim_{t\to \infty} \rho_t=\rho^*_t,
    \end{equation}
    for some initial condition, where $\rho^*_t$ does not converge to any time-independent density matrix as $t\to \infty$.
    \item[(iv)] time-dependent steady states without non-trivial time-independent ones if the conditions in classes (i), (ii), or (iii) are not satisfied.  
\end{enumerate}  
\end{dfn}

In the following sections, we discuss rigorous criteria to determine which class a GKSL dynamics  belongs to from the information of $\mathcal{L}_t$.
In Secs.~\ref{sec:review} and \ref{sec:criteria_for_unique} we mainly focus on the uniqueness of steady states, i.e., the condition for class (i).
In Sec.~\ref{sec:implication}, we further provide criteria to ensure the other classes.

\section{Review: criteria on the uniqueness for time-independent generators}
\label{sec:review}
In this section, before discussing time-dependent dynamics, we provide a brief review of steady states of the GKSL equation with time-independent generators, where $H_t=H$ and $L_{m,t}=L_m$ are satisfied.
In this case, the degeneracy of the steady state $\rho^*$, which is a density matrix given by an eigenoperator of $\mathcal{L}$
with eigenvalue $0$, is determined by symmetry of systems. When $L_m$ is Hermitian, the following relation holds~\cite{wolf_QuantumChannelsOperations_2012}:
\begin{align}
\mathcal{L}(O)=0
&\iff [H,O]=[L_m,O]=0.
\label{eq:commutant_degen}
\end{align}
The operators $O$ that satisfy the right-hand side are called the strong symmetry~\cite{buca_NoteSymmetryReductions_2012,albert_SymmetriesConservedQuantities_2014}. 
If we write the identity operator as $\mathbb{I}$, it always commutes with $H$ and $L_m$ and is regarded as a (trivial) strong symmetry. Then, Eq.~\eqref{eq:commutant_degen} indicates $\mathcal{L}(\mathbb{I})=0$, i.e., the maximally mixed state $\mathbb{I}/d$ is always the steady state. 

Let us write the set of strong symmetries as
\begin{equation}
    \mathcal{C}=\{O\in\mathcal{B}(\mathcal{H})\mid [H,O]=[L_m,O]=0\ \forall m\}.
    \label{eq:commutant_0}
\end{equation}
Here, $\mathcal{B}(\mathcal{H})$ is the algebra generated by all the operators on $\mathcal{H}$. From Eq.~\eqref{eq:commutant_degen}, $\Ker{\mathcal{L}}=\mathcal{C}$, and thus the steady state can be obtained by taking a positive semidefinite element of $\mathcal{C}$.
In particular, the following relation can be used to prove the uniqueness of the steady states:
\begin{align}
    \Ker \mathcal{L}=\{c\mathbb{I}\mid c\in \mathbb{C}\} 
    &\iff \mathcal{C}=\{c\mathbb{I}\mid c\in \mathbb{C}\}.
    \label{eq:commutant_unique}
\end{align}
In other words, the uniqueness of the maximally mixed steady state (left-hand side) is equivalent to the fact that the strong symmetry is trivial, i.e., all elements of $\mathcal{C}$ are proportional to the identity operator (right-hand side).

Next, let us rephrase the uniqueness condition \eqref{eq:commutant_unique} in terms of the algebraic structure of the Liouvillian, which is often more practically useful to verify than the criterion by strong symmetry. Let $\mathcal{A}\:(\subset\mathcal{B}(\mathcal{H}))$ be a subalgebra generated by the addition, multiplication, and scalar multiplication of $\{\mathbb{I},H,\{L_m\}\}$, which we denote as 
\begin{equation}
    \mathcal{A}=\langle \mathbb{I},H,\{L_m\} \rangle,
    \label{eq:algebra_indep}
\end{equation}
where $\{L_m\}$ describes a set composed of $L_1,\cdots, L_M$.
Then, $\mathcal{C}$ can be seen as the commutant of $\mathcal{A}$, i.e., the set of all the operators that commute with all the elements of $\mathcal{A}$. Moreover, because $\mathcal{A}$ is closed under taking adjoints, von Neumann’s bicommutant theorem~\cite{moudgalya_HilbertSpaceFragmentation_2022,moudgalya_SymmetriesCommutantAlgebras_2023,landsman_LectureNotesCalgebras_1998} tells us that $\mathcal{A}$ is the commutant of $\mathcal{C}$. In particular, the following relation holds:
\begin{align}
    \mathcal{C}=\{c\mathbb{I}\mid c\in \mathbb{C}\}\iff \mathcal{A}=\mathcal{B}(\mathcal{H}),\label{eq:bicommutant}
\end{align}
which can be demonstrated as follows. If $\mathcal{A} = \mathcal{B}(\mathcal{H})$, then any element of $\mathcal{C}$ must commute with all operators on $\mathcal{H}$. This implies that $\mathcal{C} = \{c\mathbb{I} \mid c \in \mathbb{C}\}$, since the only operators that commute with all operators on a Hilbert space are scalar multiples of the identity. Conversely, if  
$\mathcal{C} = \{c\mathbb{I} \mid c \in \mathbb{C}\}$, then $\mathcal{A}$ must be the set of all operators that commute with $\mathbb{I}$, which is simply $\mathcal{B}(\mathcal{H})$. 

By summarizing \eqref{eq:commutant_unique} and \eqref{eq:bicommutant},  we have
\begin{align}\label{equiv_timeindep}
    \Ker \mathcal{L}=\{c\mathbb{I}\mid c\in \mathbb{C}\} 
    &\iff \mathcal{C}
    =\{c\mathbb{I}\mid c\in \mathbb{C}\}\\
    &\iff \mathcal{A}
    =\mathcal{B}(\mathcal{H}).
\end{align}

\begin{ex}
 For example, consider a time-independent GKSL equation on $\mathcal{H}=\mathbb{C}^2$ given by:
\begin{equation}
    H=\frac{\omega}{2}\sigma^x,\,L=\sqrt{\kappa}\sigma^z ,
\end{equation}
where $\omega,\kappa>0$ are constants and $\sigma^\mu$ $(\mu=x,y,z)$ are Pauli operators.
Then, $\mathbb{I}$, $\sigma^x$ and $\sigma^z$ generate all the operators in $\mathcal{B}(\mathcal{H})$ (for example, $\sigma^y=i\sigma^x\sigma^z$). Thus $\mathcal{A}=\mathcal{B}(\mathcal{H})$ and  $\mathcal{C}= \{c\mathbb{I}\mid c\in \mathbb{C}\}$. Therefore, the steady state is unique and given by $\mathbb{I}/2$.   
\end{ex}

\begin{ex}
Next, we consider a GKSL equation given by:
\begin{equation}
    H=\frac{\omega}{2}\sigma^z,\,L=\sqrt{\kappa}\sigma^z ,
\end{equation}
where $\omega,\kappa>0$ are constants and $\sigma^\mu$ $(\mu=x,y,z)$ are Pauli operators.
Then $\sigma^z$ commutes with both $H$ and $L$, and we have $\mathcal{C}=\{\alpha\mathbb{I}+\beta\sigma^z\mid\alpha,\beta\in\mathbb{C}\}$. Thus $\Ker \mathcal{L}=\{\alpha\mathbb{I}+\beta\sigma^z\mid\alpha,\beta\in\mathbb{C}\}$. Since the density matrix is trace-one and positive semidefinite, steady states can be written as $\rho^*=(\mathbb{I}+a\sigma^z)/2$, where $-1\leq a\leq 1$ is a constant determined by the initial state $\rho_0$.
\end{ex}

While the above examples respectively correspond to class (i) and class (ii) in Sec.~\ref{sec:classification}, we here mention the other possibilities as well.
First, if $\mathcal{L}$ has a non-zero purely imaginary eigenvalue, it is called oscillating coherence~\cite{albert_SymmetriesConservedQuantities_2014}. In such a case, $\rho_t$ might not converge to a time-independent density matrix as $t\to\infty$. In our definition, such a case is classified in (iii) in Definition~\ref{def:steady_state}, and we say that $\rho_t$ converges to a time-dependent steady state $\rho^*_t$. Next, class (iv) in Definition~\ref{def:steady_state} is excluded for the time-independent GKSL equation, because if a time-independent $\rho^*$ that satisfies $\mathcal{L}(\rho^*)=0$ is unique, then there is no oscillating coherence~\cite{wolf_QuantumChannelsOperations_2012}.

As we have seen, in the time‑independent case, the steady state is characterized by the spectrum of the Liouvillian. However, when the Liouvillian is time‑dependent, analyzing the spectrum at each instant is not sufficient; as we will discuss below, one must carefully analyze the convergence to steady states.

\section{Criteria for unique steady states for time-dependent Liouvillians}
\label{sec:criteria_for_unique}
In this section, we rigorously identify the criteria of uniqueness of steady states, i.e., class (i) in our classification.
For this purpose, we first discuss that the naive extension of the strong symmetry for time-independent cases, which we call the strong symmetry in the Schr\"odinger picture, fails to characterize the uniqueness. 
Then, by introducing the strong symmetry in the interaction picture, we prove that the trivial strong symmetry provides a necessary and sufficient condition for the uniqueness by focusing on time-recurrent GKSL equations.
Moreover, we present a criterion for the uniqueness with respect to the algebraic structure of the GKSL generators, which is practically easier to verify than the criterion by the strong symmetry.
Finally, we demonstrate the usefulness of our results by concrete examples.

\subsection{Insufficiency of naive generalization for time-dependent generators}
\label{sec:insufficiency}

To judge whether a given GKSL equation has a unique steady state, one naive approach is to consider the algebra defined through the jump operators and Hamiltonian at an instantaneous time,
\begin{equation}
    \aast_t=\langle\mathbb{I}, H_t,\{L_{m,t}\} \rangle,   \label{eq:A_st}
\end{equation}
and its commutant 
\begin{equation}
    \cst_t=\{O\in\mathcal{B}(\mathcal{H})\mid[H_t,O]=[L_{m,t},O]=0\ \forall m\},
    \label{eq:commutant_st}
\end{equation}
in analogy with Eqs.~\eqref{eq:algebra_indep} and~\eqref{eq:commutant_0} in the time-independent case.
Then, as in the discussion around Eq.~\eqref{equiv_timeindep}, it holds that 
\begin{align}
    \Ker \mathcal{L}_t=\{c\mathbb{I}\mid c\in \mathbb{C}\} 
    &\iff \cst_t=\{c\mathbb{I}\mid c\in \mathbb{C}\}\\
    &\iff \aast_t=\mathcal{B}(\mathcal{H}).\label{eq:bicommutant_2}
\end{align}

Given this fact, it would be natural to define \emph{the strong symmetry in the Schr\"odinger picture} as follows:

\begin{align}
    \cst
    &=\bigcap_{t=0}^\infty\cst_t\\
    &=\{O\in\mathcal{B}(\mathcal{H})\mid [H_t,O]=[L_{m,t},O]=0\ \forall m,t\}.
    \label{eq:def_cst}
\end{align}
If $\cst\neq \{c\mathbb{I}\mid c\in \mathbb{C}\}$, the steady state is not unique. This is because we can construct a \textit{time-independent} steady state satisfying $\mathcal{L}_t({\rho}^*)=0$ for all $t$ with $\rho^*=\mathbb{I}/d+\epsilon O\neq \mathbb{I}/d$ using the non-trivial traceless operator $O$ common to all $t$ in Eq.~\eqref{eq:commutant_st}.
Here, $\epsilon$ is a sufficiently small variable such that $\rho^*$ becomes positive semidefinite.
This raises a question whether the inverse is true, i.e., if $\cst= \{c\mathbb{I}\mid c\in \mathbb{C}\}$ ensures the unique steady state.

Remarkably, even when $\cst_t= \{c\mathbb{I}\mid c\in \mathbb{C}\}$  (equivalent to $\aast_t=\mathcal{B}(\mathcal{H})$) for all $t$, which is a stronger condition than $\cst= \{c\mathbb{I}\mid c\in \mathbb{C}\}$, we find that the steady states might not be unique.
This is due to the possibility of \textit{time-dependent} ones with $\lim_{t\to \infty}\mathcal{L}_t({\rho}_t^*)\neq0$; namely, we cannot rule out the possibility of class (iv) in our classification.

To see this, we consider the following example:

\begin{ex}
\label{ex:rotating_dephasing}
Consider the GKSL equation on $\mathcal{H}=\mathbb{C}^2$ defined by:\
\begin{equation}
    H_t=\frac{\omega}{2}\sigma^z,\,L_t=\sqrt{\kappa}[\sigma^x \cos(\omega t)+\sigma^y \sin(\omega t)],
    \label{eq:example_1}
\end{equation}
where $\omega,\kappa>0$ are constants and $\sigma^\mu$ $(\mu=x,y,z)$ are Pauli operators.
In this case, $\aast_t=\mathcal{B}(\mathcal{H})$ and $\cst_t= \{c\mathbb{I}\mid c\in \mathbb{C}\}$ for all $t$ (and thus $\cst= \{c\mathbb{I}\mid c\in \mathbb{C}\}$), yet the steady states are not unique and given by:
\begin{equation}
    \rho^*_t= \frac{1}{2}\begin{pmatrix}
1 & ae^{-i\omega t} \\
ae^{i\omega t}  & 1 \\
\end{pmatrix} \label{eq:solution}
\end{equation}
where $-1 \leq a \leq 1$ is a constant determined by the initial state $\rho_0$.  
\end{ex}
This clearly demonstrates that the use of the algebras in \eqref{eq:A_st}
 and \eqref{eq:commutant_st}, which seem naive extensions of the time-independent case, is actually not satisfactory in discussing the unique steady state in class (i) for time-dependent cases.

\subsection{Interaction picture and strong symmetry}
\label{sec:interaction_picture}

To see why the above type of degeneracy appears, we consider the rotating frame
\begin{align}
    U_t=\mathcal{T}e^{-i\int_0^t H_{t^{\prime}}dt^\prime},\,\tilde{\rho}_t=U_t^\dagger \rho_t U_t,\ \tilde{L}_{m,t}=U_t^\dagger L_{m,t} U_t,
    \label{eq:transform_int}
\end{align}
Then, the GKSL equation \eqref{eq:GKSL} reads
\begin{align}
  \frac{d \tilde{\rho}_t}{d t}
  =\mathcal{\tilde{L}}_t (\tilde{\rho}_t)
  =\sum_{m}
  \left(\tilde{L}_{m,t} \tilde{\rho}_t
  \tilde{L}_{m,t}-\frac{1}{2}\left\{(\tilde{L}_{m,t})^{2}, \tilde{\rho}_t\right\}\right),
  \label{eq:GKSL_int}
\end{align}
which is called the GKSL equation in the interaction picture. In this representation, the model~\eqref{eq:example_1} can be described by the time-independent generator
$\tilde{L}_t=\sqrt{\kappa}\sigma^x,$
which commutes with $\sigma^x$ itself.
Then, $\tilde{{\rho}}_t^*=\mathbb{I}/2+a\sigma^x/2$ is a fixed point of $\mathcal{\tilde{L}}_t$ for all $-1\leq a\leq 1$.
This leads to the time-dependent steady state $\rho_t^*=U_t\tilde{\rho}_tU_t^\dag$ described in \eqref{eq:solution}. 

Motivated by this observation, we newly introduce a strong symmetry distinct from $\cst$, which we call
\emph{the strong symmetry in the interaction picture}, characterized by
\begin{equation}
    \cint=\{O\in\mathcal{B}(\mathcal{H})\mid [\tilde{L}_{m,t},O]=0\ \forall m,t\}.
    \label{eq:def_cint}
\end{equation}
One immediate consequence of the non-trivial strong symmetry in the interaction picture, $\cint\neq \{c\mathbb{I}\mid c\in \mathbb{C}\}$, is that it
gives a sufficient condition for multiple steady states (see proof of Theorem~\ref{thm:main1} (\emph{2}. $\to$ \emph{1}.) in Sec.~\ref{sec:necessary_and_sufficient}).
In other words, $\cint= \{c\mathbb{I}\mid c\in \mathbb{C}\}$ is a necessary condition for the unique steady state. 

The strong symmetry in the interaction picture $\cint$ is weaker than that in the Schr\"odinger picture $\cst$, i.e., there is an inclusion relation $\cst\subseteq\cint$. As we see in Sec.~\ref{sec:existence}, if and only if $\cint\setminus\cst\neq \emptyset$, there exists a time-dependent steady state. Example~\ref{ex:rotating_dephasing} satisfies this condition, because $\cst= \{c\mathbb{I}\mid c\in \mathbb{C}\}$ and $\cint= \{\alpha\mathbb{I}+\beta\sigma^x\mid \alpha,\beta\in \mathbb{C}\}$. As we will rigorously justify later, this corresponds to class (iv) in our classification.

Interestingly, even for time-independent Liouvillians, these two symmetries can be different. We will see that this is also related to coherent and oscillatory evolution protected by dynamical strong symmetry~\cite{buca_NonstationaryCoherentQuantum_2019}, leading to class (iii) in our classification.

\medskip

\subsection{Necessary and sufficient condition for unique steady states for time-recurrent Liouvillians}
\label{sec:necessary_and_sufficient}
While we have seen that $\cint= \{c\mathbb{I}\mid c\in \mathbb{C}\}$ gives a necessary condition for the unique steady state, it is highly non-trivial that it gives a sufficient condition.
For general time-dependent Liouvillians, this is indeed not true. To see this, let us consider the following time-dependent GKSL equation:

\begin{ex}[Exponentially decaying dephasing]
\label{ex:exp_decay_dephasing}
Consider the time-dependent GKSL equation with
\begin{equation}
    H=0, L_{1,t}=\sigma^xe^{-t}, L_{2,t}=\sigma^ye^{-t}.
\end{equation}
Then, $\cint= \{c\mathbb{I}\mid c\in \mathbb{C}\}$, but the steady state is not unique. To see this, we write the initial state as 
\begin{equation}
   \rho_0=\mathbb{I}/2+c_x\sigma^x+c_y\sigma^y+c_z\sigma^z,
\end{equation}
where $c_x,c_y,c_z$ are real constants. By setting
\begin{equation}
    \langle \sigma^\mu\rangle_t=\Tr[\sigma^\mu \rho_t]\ (\mu=x,y,z),
\end{equation}
Then, they obey
\begin{align}
    \partial_t\langle \sigma^\nu\rangle_t&=-2e^{-2t}\langle \sigma^\nu\rangle_t\ (\nu=x,y),\\
    \partial_t\langle \sigma^z\rangle_t&=-4e^{-2t}\langle \sigma^z\rangle_t.
\end{align}
By solving this equation, we have
\begin{align}
    \langle \sigma^\nu\rangle_t&=c_\nu e^{e^{-2t}-1}\ (\nu=x,y),\\
    \langle \sigma^z\rangle_t&=c_ze^{2(e^{-2t}-1)}.
\end{align}
Therefore,
\begin{equation}
    \lim_{t\to \infty}\rho_t=\mathbb{I}/2+e^{-1}c_x\sigma^x+e^{-1}c_y\sigma^y+e^{-2}c_z\sigma^z,
\end{equation}
which depends on the initial state. 
\end{ex}

This example illustrates that $\cint=\{c\mathbb{I}\mid c\in \mathbb{C}\}$ does not guarantee the uniqueness of the steady state for general time-dependent Liouvillians. Still, by imposing an additional condition, it can be proved that $\cint=\{c\mathbb{I}\mid c\in \mathbb{C}\}$ gives a necessary and sufficient condition for the unique steady state, i.e., class (i) in our classification.

\begin{thm}
\label{thm:main1}
Consider the time-dependent GKSL equation~\eqref{eq:GKSL} that satisfies {Condition~\ref{con:recurrent}}. Then, the following statements are equivalent:
\begin{enumerate}
    \item $\cint=\{c\mathbb{I}\mid c\in \mathbb{C}\}$.
    \item For any initial state $\rho_0$, the solution $\rho_t$ of the GKSL equation converges to the maximally mixed state:
    \begin{equation}
     \lim_{t\to \infty}\rho_t=\mathbb{I}/d.
    \end{equation}
    as the unique steady state. 
\end{enumerate}
\end{thm}

\begin{con}
\label{con:recurrent}
The Liouvillian $\mathcal{L}_t$ is continuous, bounded, i.e., there exists a constant $M>0$ such that $\opn \mathcal{L}_t\opn_2\leq M$ for all $t$, and recurrent in time in the sense that for all $\varepsilon>0$, $\delta>0$ and $t_0\geq0$, there exists an infinite sequence of mutually disjoint time intervals $[t_n,t_n+\delta]$ $(n=0,1,\ldots)$ such that
\begin{equation}
    \opn\mathcal{L}_{t_n+t} - \mathcal{L}_{t_0+t} \opn_2<\varepsilon \quad \forall t\in [0,\delta].
    \label{eq:quasiperiodic}
\end{equation}
Here, $\opn\cdot\opn_2$ is the spectral norm of superoperators
\begin{equation}
    \opn\Phi\opn_2=\max_{\|A\|_2=1} \|\Phi (A)\|_2,
    \label{eq:spectral_norm}
\end{equation}
induced by the Hilbert-Schmidt norm of operators
\begin{equation}
\|A\|_2=\sqrt{\Tr[A^\dagger A]}. 
\label{eq:hilbert_schmidt_norm}
\end{equation}
\end{con}

See Fig.~\ref{fig:condition3} for schematic illustration of Condition~\ref{con:recurrent}. The first half of Condition~\ref{con:recurrent} states that the Hamiltonian and jump operators do not diverge in time, and this is a physically natural assumption. This latter half is automatically satisfied when the Liouvillian is time-independent or time-periodic. Furthermore, it includes time-quasiperiodic systems subjected to external multi-frequency drives~\cite{ho_SemiclassicalManymodeFloquet_1983,casati_AndersonTransitionOneDimensional_1989,martin_TopologicalFrequencyConversion_2017,malz_TopologicalTwodimensionalFloquet_2021,beatrez_ObservationCriticalPrethermal_2023,park_ControllableFloquetEdge_2023,roy_SingleMultifrequencyDriving_2024,he_ExperimentalRealizationDiscrete_2025}, Fibonacci drives~\cite{maity_FibonacciSteadyStates_2019,lapierre_FineStructureHeating_2020,pilatowsky-cameo_CompleteHilbertSpaceErgodicity_2023}, Thue-Morse drives~\cite{mori_RigorousBoundsHeating_2021,zhao_RandomMultipolarDriving_2021,das_PeriodicallyQuasiperiodicallyDriven_2023,yan_PrethermalizationAperiodicallyKicked_2024,tiwari_PeriodicallyAperiodicallyThueMorse_2025,pilatowsky-cameo_CriticallySlowHilbertSpace_2025}, and even a certain class of random drive, as explained below.

Note that, to the best of our knowledge, these examples present the first formulation of a time-quasiperiodically driven GKSL equation.

\begin{figure}
 \centering
 \includegraphics[width=\linewidth]{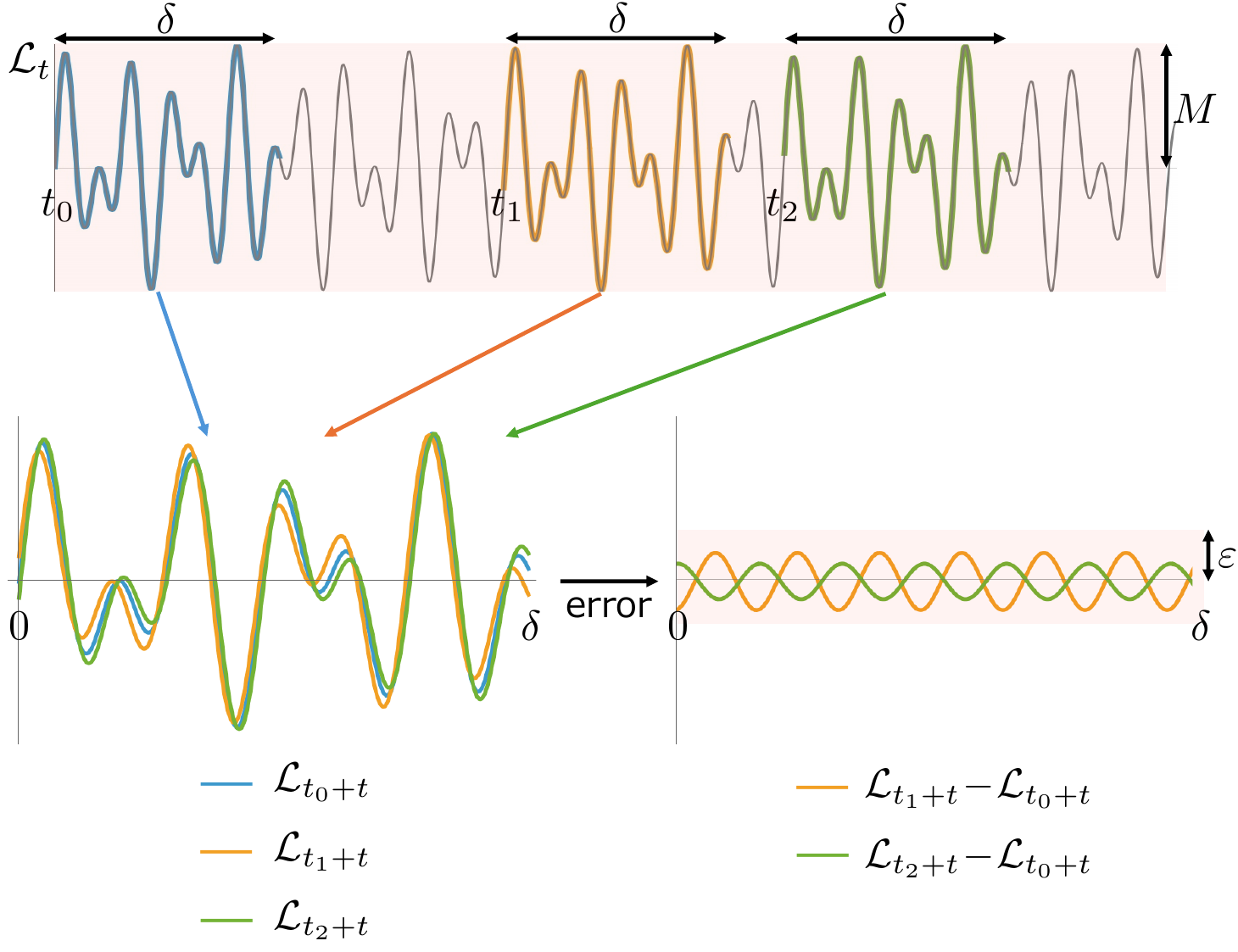}
  \caption{Schematics of $\mathcal{L}_t$ that satisfies Condition~\ref{con:recurrent}. Although $\mathcal{L}_t$ is actually a superoperator, it is depicted here as if it were a real-valued function. (top) The spectral norm of $\mathcal{L}_t$ is bounded by $M$ for any $t$. Moreover, for any $\varepsilon>0$, $\delta>0$, and $t_0\geq0$, one can take an infinite sequence of mutually disjoint time intervals $[t_n,t_n+\delta]$ $(n=0,1,\ldots)$. (bottom left) Plots of $\mathcal{L}_{t_n+t}$ $(0\leq t\leq \delta)$. (bottom right) Plot of the errors $\mathcal{L}_{t_n+t} - \mathcal{L}_{t_0+t}$. For any $n$ and $0 \leq t \leq \delta$, the spectral norms of them are bounded by $\varepsilon$.
}
\label{fig:condition3}
\end{figure}

\medskip

\begin{ex}[Multi-frequency drives]
\label{ex:multi_frequency_drives}
An example that satisfies {Condition~\ref{con:recurrent}} is a system subjected to an external multifrequency drive, where the Hamiltonian and jump operators read
\begin{align}
    H_t&=\sum_{j=1}^{l_0}\cos(\omega_{0,j} t+\phi_{0,j})H^{(j)},\\ L_{m,t}&=\sum_{j=1}^{l_m}\cos(\omega_{m,j} t+\phi_{m,j})L_{m}^{(j)}.
\end{align}
Here, $\omega_{m,j},{\phi_{m,j}}\in \mathbb{R}$ and $l_m\in \mathbb{N}$ for all $m=0,1,\ldots$. In Appendix~\ref{sec:multi_freq_proof}, we prove that this example satisfies {Condition~\ref{con:recurrent}}.
\end{ex}

\medskip

\begin{ex}[Fibonacci, Thue–Morse,
and random drives]
\label{ex:recurrent_drives}
Another example that satisfies {Condition~\ref{con:recurrent}} is a system subjected to drives generated by recurrent words, e.g., Fibonacci, Thue–Morse, and random words. First, we consider the Liouvillian given by
\begin{equation}
    \mathcal{L}^{\text{pre}}_t=
    \begin{cases}
    \mathcal{L}^{(0)} & \text{if } W_{\lfloor {t/T} \rfloor}=0, \\
    \mathcal{L}^{(1)} & \text{if } W_{\lfloor {t/T} \rfloor}=1. 
    \end{cases}
\end{equation}
where $\mathcal{L}^{(i)}$ $(i=0,1)$ are time-independent Liouvillians, ${\lfloor \cdot \rfloor}$ is a floor function, and $\{W_i\}_{i\in \{0,1,2,\ldots\}}$ is a recurrent infinite binary word where every subword appears infinitely many times. For example, the Fibonacci word is constructed recursively via string concatenations of the previous two words:
$
F^{(0)}=0,\, F^{(1)}=01,\, F^{(n)}=F^{(n-1)}+F^{(n-2)}\,(n\geq2),\, 
F =\lim_{n\to \infty}F^{(n)}=0100101001001\cdots.
$ 
As another example, the Thue-Morse word is an infinite word defined by iterating a substitution.
Let $\mu$ be the substitution rule given by
$\mu(0)=01$, $\mu(1)=10$. Starting from $T^{(0)}=0$, define inductively
$T^{(n)}=\mu\bigl(T^{(n-1)}\bigr)$ $(n\ge 1)$, $T=\lim_{n\to\infty} T^{(n)} = 011010011001\cdots$.
Finally, a random infinite binary word generated by independent coin tosses is almost surely recurrent. This follows from standard results in probability theory often referred to as the infinite monkey theorem: a random infinite sequence over a finite alphabet almost surely contains every finite string infinitely often~\cite{borel_MecaniqueStatiqueLirreversibilite_1913}.

By construction, the Liouvillian $\mathcal{L}^{\mathrm{pre}}_t$ is bounded and recurrent. However, $\mathcal{L}^{\mathrm{pre}}_t$ is not continuous in $t$. To obtain a continuous Liouvillian, we introduce a coarse‑grained version. For any arbitrarily small $a>0$, the coarse-grained Liouvillian defined by
\begin{equation}
    \mathcal{L}_t := \frac{1}{a}\int_{t}^{t+a} \mathcal{L}^{\mathrm{pre}}_\tau \, d\tau
    \label{eq:coarse‑graining}
\end{equation}
is continuous in $t$ and satisfies Condition~\ref{con:recurrent}.
\end{ex}

\begin{proof}[Proof of Theorem~\ref{thm:main1}]
(\emph{2}. $\to$ \emph{1}.)  
If $\cint\neq \{c\mathbb{I}\mid c\in \mathbb{C}\}$, there exists a nonzero, Hermitian, traceless operator $\tilde{\pi}\in\cint$. 
Since $\tilde{\pi}$ satisfies $\tilde{\mathcal{L}}_t(\tilde{\pi})=0$ for all $t$, $\tilde{\pi}$ is a time-independent solution of the differential equation~\eqref{eq:GKSL_int}. Thus, by choosing the initial condition $\tilde{\rho}_0=\mathbb{I}/d+\epsilon\tilde{\pi}$ for a sufficiently small real constant $\epsilon$, we ensure that $\tilde{\rho}_0$ is a density matrix. Then $\lim_{t\to \infty}\tilde{\rho}_t=\mathbb{I}/d+\epsilon\tilde{\pi}$, and consequently $\lim_{t\to \infty}\rho_t=\mathbb{I}/d+\epsilon U_t \tilde{\pi} U_t^\dagger \neq \mathbb{I}/d$.

(\emph{1}. $\to$ \emph{2}.)  
    Assuming that $\cint=\{c\mathbb{I}\mid c\in\mathbb{C}\}$, we prove that $\lim_{t\to\infty}\rho_t=\mathbb{I}/d$ for any $\rho_0$. By writing $\rho_t=\mathbb{I}/d+\pi_t$, where $\pi_t$ is a traceless and Hermitian operator, it suffices to prove that $\lim_{t\to\infty}\pi_t=0$ for any initial $\pi_0$.

    To prove it, we use the Hilbert-Schmidt norm defined as~\eqref{eq:hilbert_schmidt_norm}, which is related to the purity of the density matrix 
    by $\Tr[\rho_t^2]=(\|\rho_t\|_2)^2$. Then, for any GKSL equations described by Hermitian jump operators, it does not increase in time, because 
    \begin{equation}
        \partial_t\Tr[\rho_t^2]=-\sum_{m=1}^M\Tr([L_{m,t},\rho_t]^2)\leq0.
        \label{eq:purity_monotone}
    \end{equation}
    The equality holds if and only if $[L_{m,t},\rho_t]=0$ for all $m$. Since $\partial_t\Tr[\pi_t^2]=\partial_t\Tr[\rho_t^2]$,
    equation~\eqref{eq:purity_monotone} implies that $\|\pi_t\|_2$ does not increase in time, and utilizing this fact, we will prove that $\lim_{t\to\infty}\|\pi_t\|_2=0$. 

    From now on, it is convenient to use the spectral norm defined as~\eqref{eq:spectral_norm}. 
    By writing the time-evolution superoperator as $\mathcal{V}_{s,t}= \mathcal{T}e^{\int_{s}^{t}\mathcal{L}_{t'}dt'}$, equation \eqref{eq:purity_monotone} implies
     \begin{equation}
         \opn\mathcal{V}_{s,t}\opn_2\leq1\quad \forall s,t \text{ such that }0\leq s\leq t.
     \end{equation}
     If $\Phi$ is trace preserving, $\Tr A=0$ implies $\Tr \Phi (A)=0$. By restricting the domain of $\Phi$ to those with $\Tr A=0$, we define the following:
    \begin{equation}
        \opn\Phi\opn_2^\prime=\max_{\|A\|_2=1, \Tr A=0} \|\Phi (A)\|_2. 
    \end{equation}
    By definition, we have 
    \begin{equation}
        \opn\Phi\opn_2^\prime\leq\opn\Phi\opn_2,
        \label{eq:inequality_prime}
    \end{equation}
    and therefore
    \begin{equation}
         \opn\mathcal{V}_{s,t}\opn^\prime_2\leq1\quad \forall s,t \text{ such that }0\leq s\leq t.
         \label{eq:bound_v}
    \end{equation}
    Furthermore, we have
    \begin{equation}
        \opn\Phi_1\Phi_2\opn_2^\prime \leq\opn\Phi_1\opn_2^\prime \opn\Phi_2\opn_2^\prime.
        \label{eq:bound_phi1phi2}
    \end{equation}
    Moreover, if we assume Condition~\ref{con:recurrent} and $\cint=\{c\mathbb{I}\mid c\in\mathbb{C}\}$, we can prove that there exists an infinite sequence of mutually disjoint time intervals $[t_n,t_n+\delta]$ $(n=0,1,\ldots)$ such that
    \begin{equation}
        \lambda_n:=\opn\mathcal{V}_{t_n,t_n+\delta}\opn_2^\prime<\Lambda
        \label{eq:bound_lambdan_Lambda}
    \end{equation}
    for some fixed $\Lambda<1$. To prove it, we use the following lemmas (See Appendices~\ref{sec:main1_lemma1} and \ref{sec:main1_lemma2} for proof.):
    \begin{lem}
    \label{lem:lambda0_bound}
        If $\mathcal{L}_t$ is continuous and $\cint=\{c\mathbb{I}\mid c\in\mathbb{C}\}$, there exists $t_0\geq0$ and $\delta>0$ such that
        \begin{equation}
             \lambda_0<1.
        \end{equation}
    \end{lem}

    \begin{lem}
    \label{lem:lambdan_bound}
       Under Condition~\ref{con:recurrent}, for any $\varepsilon>0$, $\delta>0$, and $t_0\geq0$, there exists an infinite sequence of mutually disjoint time intervals $[t_n,t_n+\delta]$ $(n=1,2,\ldots)$ such that
       \begin{equation}
       |\lambda_n-\lambda_0|<\varepsilon\delta e^{M\delta}.
       \end{equation}
    \end{lem}
    From Lemmas \ref{lem:lambda0_bound} and \ref{lem:lambdan_bound}, by choosing sufficiently small $\varepsilon$, we have \eqref{eq:bound_lambdan_Lambda}.
    Finally, from \eqref{eq:bound_v}, \eqref{eq:bound_phi1phi2}, and \eqref{eq:bound_lambdan_Lambda}, if $t\geq t_N+\delta$ $(N=0,1,\ldots)$,
    \begin{align} 
        \opn\mathcal{V}_{0,t}\opn_2^\prime
        \leq& \opn\mathcal{V}_{0,t_0}\opn_2^\prime\left[\prod_{n=0}^{N-1} \opn\mathcal{V}_{t_n,t_n+\delta}\opn_2^\prime\opn\mathcal{V}_{t_n+\delta,t_{n+1}}\opn_2^\prime\right]\nonumber\\
        &\times\opn\mathcal{V}_{t_{N},t_{N}+\delta}\opn_2^\prime \opn\mathcal{V}_{t_{N}+\delta,t}\opn_2^\prime \nonumber\\
        \leq& \prod_{n=0}^N \opn\mathcal{V}_{t_n,t_n+\delta}\opn_2^\prime
        = \prod_{n=0}^N \lambda_n
        \leq\Lambda^{N+1}.
        \end{align}
    Since $\lim_{N\to \infty}\Lambda^{N+1}=0$,
    \begin{equation}
        \lim_{t\to \infty}\opn\mathcal{V}_{0,t}\opn_2^\prime=0,
    \end{equation}
    which guarantees that $\lim_{t\to\infty}\|\pi_t\|_2=0$ for any initial $\pi_0$. Thus, 
    \begin{equation}
        \lim_{t\to \infty}\rho_t=\mathbb{I}/d
    \end{equation}
    for any initial state $\rho_0$.
    \end{proof}
    
\subsection{Useful criterion for the uniqueness via algebra of GKSL generators}
\label{sec:useful_criterion}

While Theorem~\ref{thm:main1} provides a necessary and sufficient condition for the uniqueness of steady states, one needs to calculate $\tilde{L}_{m,t}$ to verify $\cint=\{c\mathbb{I}\mid c\in\mathbb{C}\}$, which can be technically demanding. Therefore, having a criterion formulated in the Schr\"odinger picture, which corresponds to the algebraic one based on $\mathcal{A}$ in \eqref{eq:algebra_indep} for the time-independent case, would be practically useful.
Note that $\mathcal{A}^\mathrm{Sch}_t$ in \eqref{eq:A_st} is not appropriate for justifying the uniqueness, as we saw before.

For this purpose, assuming that $\tilde{L}_{m,t}$ is smooth at $t$, we first define the following symmetry:

\begin{align}
    \tilde{\mathcal{C}}_{t}^{\mathrm{ad}}=\{O\in \mathcal{B}(\mathcal{H})\mid[\partial^n_{t}\tilde{L}_{m,t},O]=0 \ \forall m,n\}.
\end{align}
It is straightforward to check that $\tilde{\mathcal{C}}_{t}^{\mathrm{ad}}\supseteq \cint$ (see the discussion around Eq.~\eqref{eq:proof_main2_5}). Moreover, let us additionally assume that $\tilde{L}_{m,t}$ is analytical at all $t=t_0$, i.e., $\tilde{L}_{m,t}$ admits a Taylor expansion: 
\begin{equation} 
\tilde{L}_{m,t} = \sum_{n=0}^\infty \frac{(t - t_0)^n}{n!}  \partial_t^n \tilde{L}_{m,t}|_{t=t_0}
\label{eq:taylor} 
\end{equation}
in some open neighborhood $U(t_0)$ of $t=t_0$.
Then, we can prove that $\tilde{\mathcal{C}}_{t}^{\mathrm{ad}}\subseteq \cint$ (see Appendix~\ref{sec:proof_main2}), which leads to $\tilde{\mathcal{C}}_{t}^{\mathrm{ad}}=\cint$. 

Furthermore, in the Schr\"odinger picture, we define
\begin{align}
    \mathcal{C}_{t}^{\mathrm{ad}}=\{O\in \mathcal{B}(\mathcal{H})\mid [\ad^n_{t}(L_{m,t}),O]=0 \ \forall m,n\},
\end{align}
where $\ad_{t}$ is adjoint operation defined as
\begin{equation}    
    \ad_{t}(A):=i[H_{t},A]+\partial_t A.
\end{equation}
Two symmetries $\mathcal{C}_{t}^{\mathrm{ad}}$ and $\tilde{\mathcal{C}}_{t}^{\mathrm{ad}}$ are related by a unitary transformation
\begin{equation}
   \mathcal{C}_{t}^{\mathrm{ad}}=U_t\tilde{\mathcal{C}}_{t}^{\mathrm{ad}} U_t^\dagger,
   \label{eq:unitary_cad}
\end{equation}
where $U_t$ is defined as \eqref{eq:transform_int}. 

Finally, we define the subalgebra of $\mathcal{B(\mathcal{H})}$ generated by the identity operator, jump operators, and operators obtained by applying $\ad_{t}$ arbitrary finite times to jump operators, namely, 
\begin{align}
\mathcal{A}_{t}^{\mathrm{ad}}=\langle\,\mathbb{I},\{\ad^n_{t}(L_{m,t})\}_{m=1,\ldots,M;n=0,1,\ldots}\rangle.
\label{eq:aad}
\end{align}
Then, $\mathcal{C}_{t}^{\mathrm{ad}}$ is the commutant of $\mathcal{A}_{t}^{\mathrm{ad}}$. By studying $\mathcal{A}_{t}^{\mathrm{ad}}$, which is much easier to calculate than $\cint$, we can prove the uniqueness of steady states [(i) in our classification].

Under these definitions, our results are summarized as follows:

\begin{thm}
\label{thm:main2}
Consider the time-dependent GKSL equation~\eqref{eq:GKSL} that satisfies {Condition~\ref{con:recurrent}}. Suppose there exists $t_0\geq0$ such that $H_{t}$ and $L_{m,t}$ are smooth at $t=t_0$ and $\mathcal{A}_{t_0}^{\mathrm{ad}}=\mathcal{B}(\mathcal
{H})$, or equivalently, $\mathcal{C}_{t_0}^{\mathrm{ad}}=\{c\mathbb{I}\mid c\in \mathbb{C}\}$. Then, for any initial state $\rho_0$, the solution $\rho_t$ of the GKSL equation converges to the maximally mixed state: 
\begin{equation}
     \lim_{t\to \infty}\rho_t=\mathbb{I}/d
\end{equation}
as the unique steady state. Furthermore, if $H_{t}$ and $L_{m,t}$ are analytical for every $t$, then the converse is also true, i.e., the steady state is unique if and only if $\mathcal{A}_{t_0}^{\mathrm{ad}}=\mathcal{B}(\mathcal
{H})$, or equivalently, $\mathcal{C}_{t_0}^{\mathrm{ad}}=\{c\mathbb{I}\mid c\in \mathbb{C}\}$ for some $t_0$.
\end{thm}
A full proof is presented in Appendix \ref{sec:proof_main2}; only a brief outline is provided here.  From Theorem~\ref{thm:main1}, the steady state is unique if and only if $\cint=\{c\mathbb{I}\mid c\in\mathbb{C}\}$. Thus, it suffices to prove that 
\begin{equation}
\mathcal{C}_{t_0}^{\mathrm{ad}} =\{c\mathbb{I}\mid c\in\mathbb{C}\} \Longrightarrow
  \cint=\{c\mathbb{I}\mid c\in\mathbb{C}\}. 
\end{equation}
if $H_{t}$ and $L_{m,t}$ are smooth at $t=t_0$ and 
\begin{equation}
\cint=\{c\mathbb{I}\mid c\in\mathbb{C}\} \Longrightarrow\mathcal{C}_{t_0}^{\mathrm{ad}} =\{c\mathbb{I}\mid c\in\mathbb{C}\}
\end{equation}
for some $t_0$ if $H_{t}$ and $L_{m,t}$ are analytical for every $t$. 

We remark that the analyticity of $H_{t}$ and $L_{m,t}$ is essential in the converse part of Theorem~\ref{thm:main2}. In Appendix~\ref{sec:counterexample}, we construct a counterexample to the converse part of Theorem~\ref{thm:main2} in a non-analytical GKSL equation.

\subsection{Applications}
\label{sec:application_unique}

To demonstrate the applications of Theorem~\ref{thm:main2}, we consider examples of two-level systems and a boundary dissipated quantum spin chain.
Importantly, the uniqueness of the steady states for these systems cannot be proven by the existing methods in Refs.~\cite{menczel_LimitCyclesPeriodically_2019,meglio_AsymptoticRelaxationQuantum_2024}, which cannot use the information of the Hamiltonian part of the dynamics.

\begin{ex}
First, we consider a two-level system with $\mathcal{H}=\mathbb{C}^2$, whose GKSL equation is defined through
\begin{equation}
    H_t=\sigma^z+\sum_{j=1}^{l}B_j\cos(\omega_j t)\sigma^x,\,L_t=\sqrt{\kappa}\sigma^z,
    \label{eq:example_2}
\end{equation}
where $l=1,\ldots$ and $B_j,\omega_j,\kappa>0$ are constants. This GKSL equation satisfies Condition~\ref{con:recurrent}. If $l=1$, the Liouvillian is periodic in time; if $l>1$ and the ratio of the frequency is irrational, it is quasiperiodic in time.  

Then, by Theorem~\ref{thm:main2}, the steady state is unique, because
\begin{align}
    L_t|_{t=0}\propto \sigma^z, \ \ad_t[L_t]|_{t=0}\propto\sigma^y,
\end{align}
and therefore the set of operators $\{\mathbb{I},L_t|_{t=0},\ad_t[L_t]|_{t=0}\}$  generates the full operator algebra $\mathcal{B}(\mathbb{C}^2)$.  
\end{ex}

\begin{ex}
Next, consider the GKSL equation defined through
\begin{equation}
    H_t=0,\,L_t=\sqrt{\kappa}[\sigma^x \cos(\omega t)+\sigma^y \sin(\omega t)].
\end{equation}
where $\omega, \kappa>0$ are constants. Then, the steady state is unique from Theorem~\ref{thm:main2}, because
\begin{align}
    L_t|_{t=0}\propto\sigma^x,\ \ad_t[L_t]|_{t=0}\propto \sigma^y,
\end{align}
and therefore the set of operators $\{\mathbb{I},L_t|_{t=0},\ad_t[L_t]|_{t=0}\}$ generates all the full operator algebra $\mathcal{B}(\mathbb{C}^2)$.
\end{ex}

\begin{ex}
Finally, we consider an example of a boundary-dissipated quantum spin chain described by $\mathcal{H}=(\mathbb{C}^{2})^{\otimes V}$. 
Using Pauli operators, we define operators on $\mathcal{H}=(\mathbb{C}^{2})^{\otimes V}$. For $\mu=x,y,z$ and $j=1,\ldots,V$, we define
\begin{equation}
\sigma^\mu_j= \mathbb{I}\otimes \cdots \otimes\sigma^\mu\otimes \cdots \otimes\mathbb{I}
\end{equation}
where $\sigma^\mu$ acts on $j$th Hilbert space. 
Then, consider the GKSL equation defined through
\begin{align}
    H_t=\sum_{j=1}^{V-1}\sigma^z_j \sigma^z_{j+1}+B\cos(\omega t)\sum_{j=1}^{V}\sigma^x_j,\ L_t=\sqrt{\kappa}\sigma^z_1,
    \label{eq:example_3}
\end{align}
where $B,\omega,\kappa>0$ are constants and $V$ denotes the number of sites. 

We can show that this periodically driven system with boundary dephasing has a unique steady state.
Note that this also indicates the uniqueness if we instead consider the case of bulk dephasing, where the jump operator reads $L_{j,t}=\sqrt{\kappa}\sigma_j^z$.

To facilitate the proof of the uniqueness, it is convenient to introduce the Majorana operators defined by
\begin{align}
\gamma_{2j-1}=\left(\prod_{k=1}^{j-1}\sigma^x_k\right)\sigma^z_j,\ \gamma_{2j}=\left(\prod_{k=1}^{j-1}\sigma^x_k\right)\sigma^y_j,
\end{align}
which satisfy the relations
\begin{equation}
  \gamma_{j}^\dagger=\gamma_j,\  \{\gamma_{j},\gamma_{k}\}=2\delta_{jk}.
\end{equation}
Using these operators, equation~\eqref{eq:example_3} can be rewritten as
\begin{align}
    H_t=iB\cos(\omega t)\sum_{j=1}^{V}\gamma_{2j-1}\gamma_{2j}+i\sum_{j=1}^{V-1}\gamma_{2j}\gamma_{2j+1},\ L_t=\sqrt{\kappa}\gamma_1.
    \label{eq:example_3_majorana}
\end{align}
In Appendix~\ref{sec:proof_lem7}, we prove the following lemma, which ensures the uniqueness of the steady state for arbitrary $V$ through our {Theorem~\ref{thm:main2}}.
\end{ex}

\begin{lem}
\label{lem:tfim2}
Consider the GKSL equation defined by Eq.~\eqref{eq:example_3_majorana}.
For $n=0,\ldots, 2V-1$, the set of operators $\{\mathbb{I},\ad_t^{m}(L_t)|_{t=0}\mid m=0,1,\ldots,n\}$ generates all polynomials of $\gamma_1,\ldots,\gamma_{n+1}$ under addition, multiplication, and scalar multiplication.
\end{lem}

\section{Implication of strong symmetries $\cst$ and $\cint$ for steady-state structure}
\label{sec:implication}

So far, we have discussed the criteria for the unique steady states, i.e., class (i) in our classification introduced in 
Sec.~\ref{sec:classification}.
In this section, we rigorously show that the two distinct strong symmetries $\cst$ and $\cint$, introduced in Sec.~\ref{sec:criteria_for_unique}, completely characterize the four classification 
in 
Sec.~\ref{sec:classification}
for time-dependent GKSL equations.
Moreover, providing concrete examples for the four classes, we elucidate how the previously proposed mechanisms underlying non-trivial oscillatory steady states, such as strong dynamical symmetry and Floquet dynamical symmetry, are naturally understood in our framework.
Furthermore, we demonstrate a new type of symmetry-predicted, time-dependent steady-state dynamics driven by time-quasiperiodic driving, which cannot be explained by the previously known mechanisms.

\subsection{Absence of multiple time-independent steady states under $\cst=\{c\mathbb{I}\mid c\in\mathbb{C}\}$}
\label{sec:absence}
As mentioned in Sec.~\ref{sec:insufficiency}, if there is a strong symmetry in the Schr\"odinger picture, i.e., $\cst\neq\{c\mathbb{I}\mid c\in\mathbb{C}\}$, then the GKSL equation has multiple time-independent steady states [(ii) or (iii) in our classification]. Then, a natural question is whether $\cst=\{c\mathbb{I}\mid c\in\mathbb{C}\}$ implies the absence of multiple time-independent steady states [(i) or (iv) in our classification]. For a general time-dependent Liouvillian, this is not the case. 
In the case of Example~\ref{ex:exp_decay_dephasing}, $\cst=\{c\mathbb{I}\mid c\in\mathbb{C}\}$, but there are multiple time-independent steady states.

However, if we restrict Liouvillian to a time-recurrent one, we can prove that $\cst=\{c\mathbb{I}\mid c\in\mathbb{C}\}$ implies the absence of multiple time-independent steady states:
\begin{thm}
\label{thm:indep}
Consider the time-dependent GKSL equation~\eqref{eq:GKSL} that satisfies {Condition~\ref{con:recurrent}}. Let $\rho^*$ be a time-independent density matrix. Then, the following statements are equivalent:
\begin{enumerate}
    \item $\displaystyle \lim_{t\to \infty} \rho_t = \rho^*$ for some initial condition $\rho_0$.
    \item $\rho^* \in \cst$.
\end{enumerate}
\end{thm}
See Appendix~\ref{sec:proof_thm_indep} for proof. From Theorem~\ref{thm:indep}, $\cst=\{c\mathbb{I}\mid c\in\mathbb{C}\}$ is equivalent to the absence of multiple time-independent steady states, i.e., classes (i) and (iv) in our classification, for time-dependent GKSL equations. 

\subsection{Existence of time-dependent steady states under $\cint\setminus\cst\neq \emptyset$}
\label{sec:existence}
An interesting situation is that $\cst=\{c\mathbb{I}\mid c\in\mathbb{C}\}$ but $\cint\neq\{c\mathbb{I}\mid c\in\mathbb{C}\}$. 
In this case, from Theorems~\ref{thm:main1} and ~\ref{thm:indep}, we can expect that the GKSL equation has time-dependent steady states, as exemplified in Example~\ref{ex:rotating_dephasing}.
Here, more generally, we show that there exists a time-dependent steady state [(iii) and (iv) in our classification] if and only if $\cint\setminus\cst\neq \emptyset$.

\begin{thm}
\label{thm:time_dependent}
Consider the time-dependent GKSL equation~\eqref{eq:GKSL} that satisfies {Condition~\ref{con:recurrent}}.
Then, the following statements are equivalent:
\begin{enumerate}
    \item $\cint\setminus\cst\neq \emptyset$.
    \item There exists a time-dependent steady state, i.e., there is a time-dependent $\rho^*_t$ that does not converge to any time-independent density matrix as $t\to \infty$ such that 
    \begin{equation}
    \lim_{t\to \infty} \rho_t=\rho^*_t,
    \end{equation}
    for some initial condition $\rho_0$.
\end{enumerate}
\end{thm}

The proof is given in Appendix~\ref{sec:proof_thm_time_dependent}.
There, we use the following Lemma proven in Appendix~\ref{sec:proof_thm_cint}, which follows from Theorem~\ref{thm:main1}.

\begin{lem}
\label{lem:cint}
Let $\rho^*_t$ be a density matrix that may depend on time $t$. Under Condition~\ref{con:recurrent}, the following statements are equivalent:
\begin{enumerate}
    \item $\lim_{t\to \infty} \rho_t = \rho^*_t$ for some initial condition $\rho_0$.
    \item $\rho^*_t=U_t \tilde{\rho}^*U_t^\dagger$ for some time-independent density matrix $\tilde{\rho}^*\in \cint$.
\end{enumerate}
\end{lem}

From Theorems \ref{thm:indep}, and \ref{thm:time_dependent}, we can rigorously complete the characterization of the four steady-state classes (i)-(iv) introduced in Sec.~\ref{sec:classification} via
the two strong symmetries $\cst$ and $\cint$.
This is summarized as Fig.~\ref{fig:schematic} (b).

\subsection{Applications}
In this section, we discuss applications of Theorems \ref{thm:indep} and \ref{thm:time_dependent}. We first see that strong dynamical symmetry~\cite{buca_NonstationaryCoherentQuantum_2019,sanchezmunoz_SymmetriesConservationLaws_2019,dogra_DissipationinducedStructuralInstability_2019,tindall_QuantumSynchronisationEnabled_2020} and Floquet dynamical symmetry~\cite{medenjak_RigorousBoundsDynamical_2020,chinzei_TimeCrystalsProtected_2020}, which are proposed as mechanisms for the long-time oscillations, can be understood from our framework based on $\cst$ and $\cint$. 
Moreover, we propose novel systems where an asymptotic time-dependent steady state due to quasiperiodic driving is protected by the symmetry structure $\cint\setminus\cst\neq \emptyset$.

\subsubsection{Strong dynamical symmetry}
\label{sec:strong_dynamical_symmetry}
To see a situation such that  $\cint\setminus\cst\neq \emptyset$ holds, we first consider a time-independent Liouvillian $\mathcal{L}$ and focus on  a symmetry $A$ that satisfies the following condition:
\begin{equation}
[H,A]=\Omega A \text{ and } [L_m,A]=0\ \forall m
\label{eq:dynamical_symmetry_static}
\end{equation}
with non-zero real $\Omega$. Such a symmetry is often called a strong dynamical symmetry~\cite{buca_NonstationaryCoherentQuantum_2019}, and with this type of symmetry, the Liouvillian has purely imaginary eigenvalues that lead to oscillating coherence. In such a case, $A,A^\dagger\notin \cst$ by definition, but $A,A^\dagger\in \cint$, because
\begin{align}
    U_t AU_t^\dagger=e^{-i \Omega t}A
\end{align}
and therefore
\begin{align}
    [\tilde{L}_{m,t},A]
    &=[U_t^\dagger L_m U_t,A]\\
    &=U_t^\dagger (L_m U_t AU_t^\dagger-U_t AU_t^\dagger L_m)U_t\\
    &=e^{-i \Omega t}U_t^\dagger [L_m,A]U_t=0,\\
    [\tilde{L}_{m,t},A^\dagger]&=-[\tilde{L}_{m,t},A]^\dagger=0.
\end{align}
Also, $A^\dagger A,\,AA^\dagger\in \cst$, because 
\begin{equation}
    [H,A^\dagger]=-\Omega A^\dagger \text{ and } [L_m,A^\dagger ]=0\ \forall m
\end{equation}
and thus
\begin{equation}
    [H,A^\dagger A]=[H,A A^\dagger]=[L_m,A^\dagger A]=[L_m, AA^\dagger]=0\ \forall m.
\end{equation}
Note that either $A^\dagger A$ or $AA^\dagger$ is not an identity operator, that is, $A$ is not unitary. To see this, let $\ket{\psi}$ be an eigenstate of $H$ with eigenvalue $E$. Then $A^n \ket{\psi}$ is an eigenstate of $H$ with eigenvalue $E+n\Omega$. If $A$ were a unitary operator, it would be invertible, so $A^n\ket{\psi}$ would not vanish for any $n$. Thus, $H$ would have arbitrarily large eigenvalues, but this contradicts the boundedness of $H$. Therefore, $A^\dagger A$ or $AA^\dagger$ is not an identity operator. Therefore, strong dynamical symmetry~\eqref{eq:dynamical_symmetry_static} leads to $\cst\neq\{c\mathbb{I}\mid c\in \mathbb{C}\}$ and  $\cint\setminus\cst\neq \emptyset$. From Fig.~\ref{fig:schematic}, this implies class (iii), i.e., multiple time-independent and time-dependent steady states in the sense of {Definition~\ref{def:steady_state}}. 

\begin{ex}[\cite{buca_NonstationaryCoherentQuantum_2019}]
\label{ex:buca}
In Ref.~\cite{buca_NonstationaryCoherentQuantum_2019}, the authors proposed a dissipative Hubbard model with strong dynamical symmetry, whose Hamiltonian and jump operators are given by
\begin{align}
    H=
    &-\tau\sum_{\langle j,j^\prime\rangle}\sum_{\sigma=\uparrow,\downarrow}\left(c_{j,\sigma}^\dagger c_{j^\prime,\sigma}+\text{h.c.}\right)\nonumber\\
    &+\sum_{j}\left[Un_{j,\uparrow}n_{j,\downarrow}+\mu_j n_j+\frac{B}{2}(n_{j,\uparrow}-n_{j,\downarrow})\right]\\
    L_j=&\sqrt{\kappa_j}n_j,
\end{align}
where $j,j^\prime$ denote sites of a finite bipartite lattice, $c_{j,\sigma}^\dagger$ and $c_{j,\sigma}$ denote creation and annihilation operators of fermion on site $j$ with spin $\sigma=\uparrow,\downarrow$. The particle number operator for each spin is $n_{j,\sigma}=c_{j,\sigma}^\dagger c_{j,\sigma}$, the total particle number is  $n_{j}=n_{j,\uparrow}+n_{j,\downarrow}$, and $\tau,U,\mu_j,B,\kappa_j$ are real constants that satisfy $\tau,B\neq0$ and $\kappa_j>0$. 

The dynamical symmetry of the model is given by 
\begin{align}
S^+=\sum_{j}c_{j,\uparrow}^\dagger c_{j,\downarrow}&,\ S^-=\sum_{j}c_{j,\downarrow}^\dagger c_{j,\uparrow} \\
[H,S^\pm]=\pm BS^\pm&,\ [L_j,S^\pm]=0.
\label{eq:hubbard_indep_dynamical}
\end{align}
Also, the model has $U(1)\times U(1)$ strong symmetry:
\begin{align}
N=\sum_{j}n_j&,\ S^z=\sum_{j}(n_{j,\uparrow}-n_{j,\downarrow}), \\
[H,N]=[L_j,N]&=[H,S^z]=[L_j,S^z]=0.
\label{eq:hubbard_indep_static}
\end{align}
From the above, we identify two strong symmetries $\cst$ and $\cint$. First, from \eqref{eq:hubbard_indep_dynamical}, we have $\langle\mathbb{I},S^+S^-,S^-S^+\rangle\subseteq\cst$ and $\langle\mathbb{I},S^+,S^-\rangle\subseteq\cint$. From the commutation relation $[S^+,S^-]=2S^z$, we find
$\langle\mathbb{I},S^+S^-,S^-S^+\rangle=\langle\mathbb{I},S^+S^-,S^z\rangle$, and $\langle\mathbb{I},S^+,S^-\rangle=\langle\mathbb{I},S^+,S^-,S^z\rangle$.
By considering \eqref{eq:hubbard_indep_static}, we have
\begin{align}
    \langle\mathbb{I},N,S^+S^-,S^z\rangle&\subseteq\cst,\\
    \langle\mathbb{I},N,S^+,S^-,S^z\rangle&\subseteq\cint,\\
    \operatorname{span}\{(S^+)^m (S^-)^n\langle \mathbb{I},N,S^z\rangle\mid m\neq n\}&\subseteq\cint\setminus\cst.
\end{align}
Owing to $\cint\setminus\cst\neq0$, there is a time-dependent steady state in our Definition~\ref{def:steady_state}, which is consistent with the existence of oscillating coherence protected by strong dynamical symmetry discussed in~\cite{buca_NonstationaryCoherentQuantum_2019}.
Note that since $\cst\neq\{c\mathbb{I}\mid c\in \mathbb{C}\}$, there exist multiple time-independent steady states as well.
In fact, multiple time-dependent steady states always accompany multiple time-independent ones for time-independent GKSL equations~\cite{wolf_QuantumChannelsOperations_2012}, leading to class (iii) in our classification.
\end{ex}

\subsubsection{Floquet dynamical symmetry}
\label{sec:floquet_dynamical_symmetry}
The counterpart of the strong dynamical symmetry in a time-periodic Liouvillian is called the Floquet dynamical symmetry~\cite{chinzei_TimeCrystalsProtected_2020,medenjak_RigorousBoundsDynamical_2020}. For a time-periodic Hamiltonian $H_{t}=H_{t+T}$ and time-indepedent jump operators $\{L_m\}$, $A$ is called Floquet dynamical symmetry if
\begin{align}
    U_T A U_T^\dagger&=e^{-i\Omega T}A\quad (\Omega\neq0)\\
    [L_m,U_t AU_t^\dagger]&=0 \quad \forall m,t,
    \label{eq:floquet dynamical symmetry}
\end{align}
where $U_t=\mathcal{T}e^{-i\int_0^t H_{t^{\prime}}dt^\prime}$ (see also \eqref{eq:transform_int}). Then, $A,A^\dagger\in \cint$, bacause
\begin{align}
    [\tilde{L}_{m,t},A]
    &=[U_t^\dagger L_m U_t,A]\\
    &=U_t^\dagger [L_m,U_t AU_t^\dag]U_t=0,\\
    [\tilde{L}_{m,t},A^\dagger]&=-[\tilde{L}_{m,t},A]^\dagger=0.
\end{align}
by \eqref{eq:floquet dynamical symmetry}. However, $A,A^\dagger\notin \cst$, bacause if so, $\Omega=0$. Moreover, unlike the strong dynamical symmetry in the time-independent case, $A^\dagger A,\, AA^\dagger \notin \cst$ in general. Thus, it is possible that $\cst=\{c\mathbb{I}\mid c\in \mathbb{C}\}$ and  $\cint\setminus\cst\neq \emptyset$, which implies class (iv), i.e., time-dependent steady states without non-trivial time-independent ones in {Definition~\ref{def:steady_state}}.

\begin{ex}
\label{ex:chinzei}
In Ref.~\cite{chinzei_TimeCrystalsProtected_2020}, the authors proposed a dissipative Hubbard model with Floquet dynamical symmetry. They discuss the general Bose and Fermi-Hubbard model, but for simplicity, we focus on the Fermi-Hubbard model whose Hamiltonian and jump operator are given by
\begin{align}
    H_t=
    &-\tau\sum_{\langle j,j^\prime\rangle}\sum_{\sigma=\uparrow,\downarrow}\left(c_{j,\sigma}^\dagger c_{j^\prime,\sigma}+\text{h.c.}\right)\nonumber\\
    &+\sum_{j}\left[Un_{j,\uparrow}n_{j,\downarrow}+\mu_j n_j\right]\nonumber\\
    &+\frac{B}{2}\left[e^{-i\omega t}S^++e^{i\omega t}S^-\right]\\
    L_j=&\sqrt{\kappa_j}n_j,
\end{align}
where we used the notations in Example~\ref{ex:buca}. 
Note that we have slightly modified the model from that given in Ref.~\cite{chinzei_TimeCrystalsProtected_2020}; while they considered nearest-neighbor interaction $n_jn_{j^\prime}$, we instead consider a site-dependent potential $\mu_j n_j$ to align with Example~\ref{ex:buca}. The Floquet dynamical symmetry of the model is given by 
\begin{align}
[U_T,S^\pm_h]=e^{\pm i(\sqrt{\omega^2+B^2}) t}S^\pm_h&,\ [L_j,U_tS^\pm_h U_t^\dagger]=0\ \forall t,
\label{eq:hubbard_periodic_dynamical}
\end{align}
where $S^\pm_h$ denotes the spin raising and lowering operator along vector $(B,0,\omega)$~\footnote{Explicit form of $S^\pm_h$ is given by:
\begin{equation} S^+_h=\sum_{j}\sum_{\sigma,\tau=\uparrow,\downarrow}s_h^{\sigma\tau}c_{j,\sigma}^\dagger c_{j,\tau},\,S^-=(S^+_h)^\dagger
\end{equation}
where
\begin{align}
    -s_h^{\uparrow\uparrow}&=s_h^{\downarrow\downarrow}=\frac{\xi}{2\sqrt{1+\xi^2}}, \\
    s_h^{\uparrow\downarrow}&=\frac{1}{2}\left(\frac{1}{\sqrt{1+\xi^2}}+1\right),\\
    s_h^{\downarrow\uparrow}&=\frac{1}{2}\left(\frac{1}{\sqrt{1+\xi^2}}-1\right),
\end{align}
and $\xi=\frac{B}{\omega}$.
}.
Also, the model has $U(1)$ strong symmetry in the Schr\"odinger picture:
\begin{align}
[H_t,N]=[L_{j},N]=0\ \forall t.
\label{eq:hubbard_periodic_static}
\end{align}

From the above, we identify two strong symmetries $\cst$ and $\cint$. First, from \eqref{eq:hubbard_periodic_dynamical}, we have $\langle\mathbb{I},S^+_h,S^-_h\rangle\subseteq\cint$. From the commutation relation $[S^+_h,S^-_h]=S^z_h$, $\langle\mathbb{I},S^+_h,S^-_h\rangle=\langle\mathbb{I},S^+_h,S^-_h,S^z_h\rangle$, where $S^z_h$ denotes the spin $z$ operator along vector $(B,0,\omega)$. Note that $\langle\mathbb{I},S^+_h,S^-_h,S^z_h\rangle=R_h^\dagger\langle\mathbb{I},S^+,S^-,S^z\rangle R_h$, where $R_h$ denotes a spin rotation operator defined by $R_h=\operatorname{exp}[-i\theta_h S^y]$, where $\theta_h =\arctan(\frac{B}{\omega})$ and $S^y=\frac{1}{2i}(S^+-S^-)$. Since $R_h\in\langle\mathbb{I},S^+,S^-,S^z\rangle$, we have $\langle\mathbb{I},S^+_h,S^-_h,S^z_h\rangle=\langle\mathbb{I},S^+,S^-,S^z\rangle$. Thus, $\langle\mathbb{I},S^+,S^-,S^z\rangle\subseteq\cint$. However, $\langle\mathbb{I},S^+,S^-,S^z\rangle \cap \cst=\{c\mathbb{I}\mid c\in\mathbb{C}\}$, because the elements of $\cst$ should commute with $S^+$ and $S^-$, but such operators in $\langle\mathbb{I},S^+,S^-,S^z\rangle$ are limited to $\{c\mathbb{I}\mid c\in\mathbb{C}\}$.
Finally, from \eqref{eq:hubbard_periodic_static}, we have
\begin{align}
    \langle\mathbb{I},N\rangle&\subseteq\cst,\\
   \label{eq:chinzei_cst} \langle\mathbb{I},N,S^+,S^-,S^z\rangle&\subseteq\cint,\\
    \label{eq:chinzei_cint} 
    \operatorname{span}\{(S^+)^m (S^-)^n\langle \mathbb{I},N\rangle\mid mn\neq0\}&\subseteq\cint\setminus\cst.
\end{align}
Since $\cint\setminus\cst\neq0$, there is a time-dependent steady state in our Definition~\ref{def:steady_state}, which is consistent with the existence of coherent oscillation protected by Floquet dynamical symmetry discussed in~\cite{chinzei_TimeCrystalsProtected_2020}. 
We can also show that the equality in~\eqref{eq:chinzei_cst} holds when $\tau,\kappa_j,B,\omega \neq 0$ (see Appendix~\ref{sec:proof_hubbard_commutant} for the proof). This implies that, if the particle number is fixed, the Liouvillian admits time-dependent steady states but no non-trivial time-independent ones, corresponding to class (iv). This is in contrast to Example~\ref{ex:buca}, where $S^+S^-$ was a conserved quantity and multiple time-independent steady states also exist even after the particle number is fixed.

The results of the numerical simulations performed using the fourth-order Runge–Kutta method for this example are shown in Figs.~\ref{fig:hubbard_quasi}(a)--(c).
Figure~\ref{fig:hubbard_quasi}(a) shows the trajectory of the expectation value of
$
S^x_1=\frac{1}{2}(c_{1,\uparrow}^\dagger c_{1,\downarrow}+c_{1,\downarrow}^\dagger c_{1,\uparrow})
$ and 
$
S^y_1=\frac{1}{2i}(c_{1,\uparrow}^\dagger c_{1,\downarrow}-c_{1,\downarrow}^\dagger c_{1,\uparrow}),
$ 
while Figure~\ref{fig:hubbard_quasi}(b) displays the time evolution of $S^y_1$. From these results, a time-dependent steady state is clearly observed, which is consistent with the discussion above.
Figure~\ref{fig:hubbard_quasi}(c) plots the absolute value of the Fourier transform of Fig.~\ref{fig:hubbard_quasi}(b).
In this case, three distinct peaks are observed, indicating that the time evolution of $S^y_1$ is a superposition of oscillations with three different frequencies.
This implies that the steady state is quasiperiodic.
These results are consistent with those observed in Ref.~\cite{chinzei_TimeCrystalsProtected_2020}.
\end{ex}

\begin{figure*}
 \centering \includegraphics[width=\linewidth]{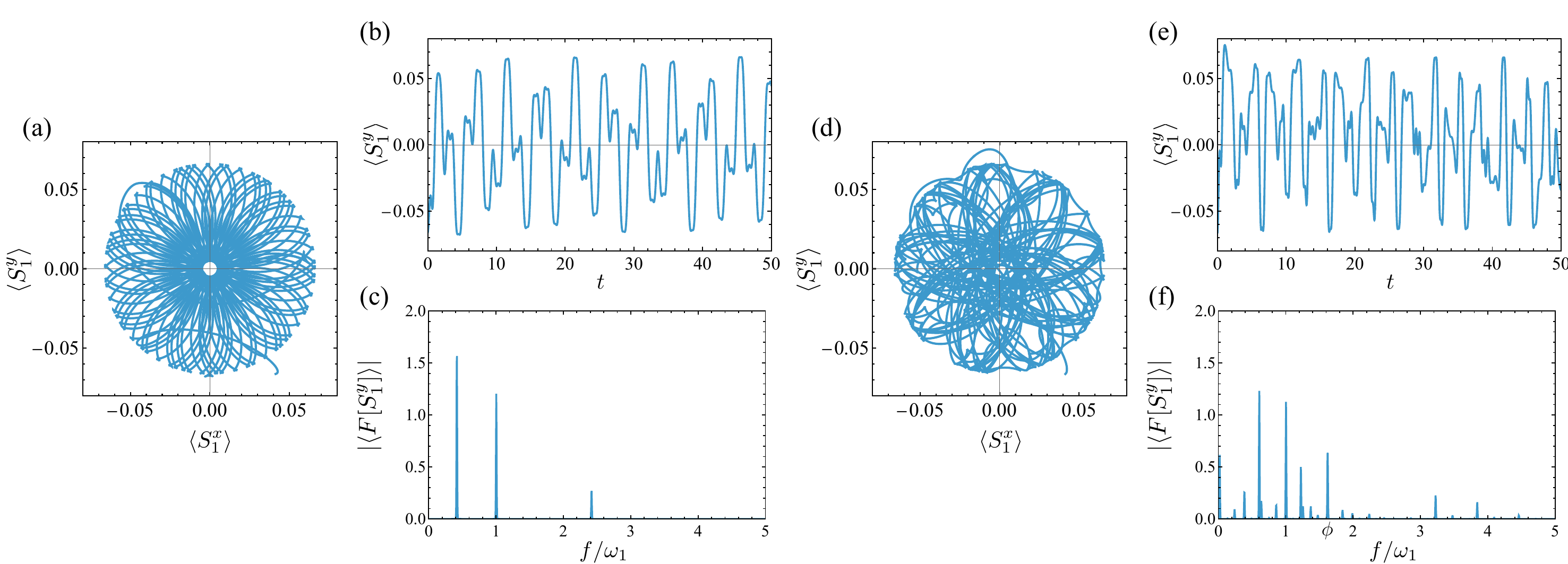}
  \caption{
Numerical simulations of the dynamics of a dissipative Hubbard model under periodic (a--c) and quasiperiodic (d--f) driving.
(a,d) Trajectories of the time evolution of the expectation values $\langle S_1^x\rangle$ and $\langle S_1^y\rangle$.
(b,e) Time evolution of $\langle S_1^y\rangle$.
(c,f) Magnitude of the Fourier spectra corresponding to (b) and (e).
The dynamics under the quasiperiodic driving leads to a much larger number of peaks than that under the periodic driving.
The Fourier analysis is performed using data in the time interval $t\in[0,500]$, with a window function
$\exp\!\left[-(t-300)^2/100^2\right]$.
The parameters are $L=4$, $J=U=\kappa_j=1$, and
$\varepsilon=\{-0.786,\,0.657,\,-0.133,\,-0.176\}$.
For periodic (quasiperiodic) driving, we set $B=\omega=\pi$ ($B_1=B_2=\omega_1=\pi$, and $\omega_2=\phi\pi$),
where $\phi=\frac{1+\sqrt{5}}{2}$. The initial state is a random pure state with fixed particle number $N=4$.
}
\label{fig:hubbard_quasi}
\end{figure*}

\subsubsection{Time-quasiperiodic case}
\label{sec:time_quasiperiodic_symmetry}
In this section, we see applications of Theorems \ref{thm:indep} and \ref{thm:time_dependent} to quasiperiodic systems, which cannot be explained by the existing approaches. The simplest example is the following toy model:

\begin{ex}
\label{ex:rorationg_dephasing_2}
Consider the GKSL equation on $\mathcal{H}=\mathbb{C}^3$ defined through
\begin{align}
H=\begin{pmatrix}
\omega_1 & 0 & 0 \\
0 & 0 & 0 \\
0 & 0 & -\omega_2
\end{pmatrix},L_t=\begin{pmatrix}
0 & e^{-i  \omega_1t } & 0 \\
e^{i \omega_1t } & 0 & e^{-i \omega_2t } \\
0 & e^{i \omega_2t } & 0
\end{pmatrix}.
\end{align}
Where $\omega_1,\omega_2>0$ are real constants. This Hamiltonian and jump operator are in the form of Example~\ref{ex:multi_frequency_drives}. Thus, when $\omega_1/\omega_2$ is irrational, the Liouvillian $\mathcal{L}_t$ is not periodic in time but satisfies Condition~\ref{con:recurrent}. 

The strong symmetry in the Schr\"odinger picture is given by
\begin{equation}
    \cst=\{c\mathbb{I}\mid c\in \mathbb{C}\},
\end{equation}
but since
\begin{equation}
    \tilde{L}_t=T^x:=\begin{pmatrix}
0 & 1 & 0 \\
1 & 0 & 1 \\
0 & 1 & 0
\end{pmatrix},
\end{equation}
the strong symmetry in the interaction picture is given by 
\begin{equation}
    \cint=\langle\mathbb{I},T^x\rangle.
\end{equation}
From Theorems~\ref{thm:indep} and \ref{thm:time_dependent}, there are time-dependent steady states and no non-trivial time-independent steady states. The steady state is explicitly given by:
\begin{equation}
    \rho^*_t=\begin{pmatrix}
\frac{1}{3}(1-b) & ae^{-i\omega_1 t} & be^{-i(\omega_1+\omega_2) t} \\
ae^{i\omega_1 t} & \frac{1}{3}(1+2b) & ae^{-i\omega_2 t} \\
be^{i(\omega_1+\omega_2) t} & ae^{i\omega_2 t} & \frac{1}{3}(1-b)
\end{pmatrix},
\end{equation}
where $a,b$ are real constants that depend on the initial state. 
\end{ex}

It is also possible to construct an example in a many-body system with a realistic setup:
\begin{ex}
We consider the Hubbard model with two-frequency drives, which can be realized with ultracold atoms on an optical lattice with two circularly polarized lasers:
\label{ex:hubbard_quasi}
    \begin{align}
    H_t=
    &-\tau\sum_{\langle j,j^\prime\rangle}\sum_{\sigma=\uparrow,\downarrow}\left(c_{j,\sigma}^\dagger c_{j^\prime,\sigma}+\text{h.c.}\right)\nonumber\\
    &+\sum_{j}\left[Un_{j,\uparrow}n_{j,\downarrow}+\mu_j n_j\right]\nonumber\\
    &+\frac{B_1}{2}\left[e^{-i\omega_1 t}S^++e^{i\omega_1 t}S^-\right]\nonumber\\
    &+\frac{B_2}{2}\left[e^{-i\omega_2 t}S^++e^{i\omega_2 t}S^-\right],\\
    L_j=&\sqrt{\kappa_j}n_j.
\end{align}
where we use the notations in Examples~\ref{ex:buca} and \ref{ex:chinzei}. When $B_1,B_2\neq0$ and $\omega_2/\omega_1$ is irrational, the Liouvillian $\mathcal{L}_t$ is not periodic in time but satisfies Condition~\ref{con:recurrent}. 

Despite $\mathcal{L}_t$ being non-periodic, we find that $\cst$ and $\cint$ are the same as in Example~\ref{ex:chinzei}:
\begin{align}
    \langle\mathbb{I},N\rangle&\subseteq\cst,\\
    \langle\mathbb{I},N,S^+,S^-,S^z\rangle&\subseteq\cint,\\
    \operatorname{span}\{(S^+)^m (S^-)^n\langle \mathbb{I},N\rangle\mid mn\neq0\}&\subseteq\cint\setminus\cst.
\end{align}
This is because, 
\begin{equation}
    [L_j,A]=0
\end{equation}
for any $A\in\langle\mathbb{I},N,S^+,S^-,S^z\rangle$, and $\langle\mathbb{I},N,S^+,S^-,S^z\rangle$ is invariant under the adjoint operation of $H_t$, i.e.,
\begin{equation}
    [H_t,\langle\mathbb{I},N,S^+,S^-,S^z\rangle]=\langle\mathbb{I},N,S^+,S^-,S^z\rangle
\end{equation}
Therefore,
\begin{align}
    [\tilde{L}_{j,t},A]=U_t^\dagger [L_{j,t},U_t A U_t^\dagger]U_t=0
\end{align}
for any $A\in\langle\mathbb{I},N,S^+,S^-,S^z\rangle$. In contrast, like Example~\ref{ex:chinzei}, $\langle\mathbb{I},S^+,S^-,S^z\rangle \cap \cst=\{c\mathbb{I}\mid c\in\mathbb{C}\}$. 

Since $\cint\setminus\cst\neq0$, there is a time-dependent steady state in our Definition~\ref{def:steady_state}. 
We can also show that the equality in~\eqref{eq:chinzei_cst} holds when $\tau,\kappa_j,B_1,B_2,\omega_1,\omega_2 \neq 0$ (see Appendix~\ref{sec:proof_hubbard_commutant} for the proof). This implies that, if the particle number is fixed, the Liouvillian admits time-dependent steady states but no non-trivial time-independent ones, corresponding to class (iv).

Since the existing methods based on strong dynamical symmetry and Floquet dynamical symmetry cannot be applied to quasiperiodic systems, this example illustrates the advantage of our approach using $\cst$ and $\cint$.

Figures~\ref{fig:hubbard_quasi}(d)--(f) summarize the numerical results for the present dynamics.
In Fig.~\ref{fig:hubbard_quasi}(d), we plot the trajectory formed by the expectation values of
$S^x_1$ and $S^y_1$, whereas Fig.~\ref{fig:hubbard_quasi}(e) shows the temporal behavior of $S^y_1$.
These plots reveal the emergence of a steady regime with explicit time dependence, in agreement with theoretical considerations based on symmetries.
Figure~\ref{fig:hubbard_quasi}(f) presents the magnitude of the Fourier spectrum corresponding to Fig.~\ref{fig:hubbard_quasi}(e).
In contrast to the previous case, a large number of peaks are observed, indicating that the evolution of the density matrix involves a superposition of oscillatory components with many distinct frequencies.

\end{ex}

\section{Conclusion and outlook}
\label{sec:conclusion}
We have established a general framework for rigorously classifying steady states of time-recurrent GKSL equations with Hermitian jump operators.
Our key idea is to distinguish two types of strong symmetry: strong symmetry in the Schr\"odinger picture $\cst$ and strong symmetry in the interaction picture $\cint$,  where a general hierarchy $\cst\subseteq\cint$ exists.

We have first shown criteria for the uniqueness of steady states.
While a straightforward generalization of strong symmetry in the time-independent GKSL equations is $\cst$, the condition $\cst=\{c\mathbb{I}\mid c\in\mathbb{C}\}$ does not guarantee the uniqueness of steady states. 
In contrast, we have proven that $\cint=\{c\mathbb{I}\mid c\in\mathbb{C}\}$ provides the equivalent condition for the unique steady state of time-recurrent GKSL equations.
Moreover, another practical criterion of the uniqueness is given in terms of the algebra $\mathcal{A}_{t}^{\mathrm{ad}}$ of the GKSL Liouvillian. When the Hamiltonian and the jump operators admit the Taylor expansion, this criterion gives a sufficient and necessary condition for the uniqueness (Fig.~\ref{fig:schematic}(a)). We have also demonstrated the application of the criterion using examples of two-level systems and a quantum spin chain, which could not be addressed by the existing methods.

Then, we have elucidated that the two distinct strong symmetries, $\cst$ and $\cint$, can completely characterize the steady-state classification given in Fig.~\ref{fig:schematic}(b).
More concretely, the existence of a non-trivial element in $\cst$ is equivalent to the existence of a non-trivial \textit{time-independent} steady state.
Moreover, the existence of an element in $\cint\setminus\cst$ is equivalent to the existence of a non-trivial \textit{time-dependent} steady state. 
We have seen that strong dynamical symmetry for time-independent GKSL equations and Floquet dynamical symmetry for time-periodic GKSL equations, which are known mechanisms to protect such a time-dependent steady state, can be understood from $\cint\setminus\cst\neq0$. Furthermore, we have constructed novel time-quasiperiodic GKSL equations, where the asymptotic time-dependent dynamics can be understood by  $\cint\setminus\cst\neq0$ but not by the conventionally known strong symmetries. The examples include, e.g., the Hubbard model with multi-frequency drives, which may be realized with ultracold atoms in optical lattices.

This work will stimulate many intriguing future questions.
One important direction is to extend the framework to cases where the jump operators are non-Hermitian. 
In such cases, one may encounter interesting scenarios, e.g., where the steady state is unique yet exhibits explicit time dependence.
Such extensions introduce significant challenges: in the proof of Theorem~\ref{thm:main1}, it was essential that the purity of the density matrix monotonically decreases in time, which is not the case for non-Hermitian jump operators. Therefore, to prove the convergence to unique steady states, alternative measures may be required to quantify the distance from the steady state.
We also suppose that the relations between the strong symmetries and the steady state are not simple in the case of non-Hermitian jump operators, as in the time-independent GKSL equations~\cite{wolf_QuantumChannelsOperations_2012,yoshida_UniquenessSteadyStates_2024,zhang_CriteriaDaviesIrreducibility_2024,hamazaki_IntroductionMonitoredQuantum_2026,fagnola_IrreducibilityQuantumMarkov_2025}.

Second, it is also important to see if our results are extended to more general time-dependent systems than time-recurrent ones. 
As we have seen in Example~\ref{ex:exp_decay_dephasing}, we can construct a counterexample of our results if arbitrary time-dependence is allowed.
While this example describes the dissipation whose strength is decaying in time, it is interesting to ask if our theorems still apply if we focus on non-decaying dissipation. 

Third, our findings on two distinct strong symmetries will open a new avenue for studying fundamental symmetries in GKSL equations. 
For instance, while we focused on strong symmetries, we can also define two distinct \textit{weak} symmetries $\cst_w$ and $\cint_w$ with respect to Schr\"odinger and interaction pictures for time-dependent GKSL equations. They may be defined from unitary operators $V$ that satisfy $V\mathcal{L}_t[\cdot]V^\dag =\mathcal{L}_t[V\cdot V^\dag]$ and $V\mathcal{\tilde{L}}_t[\cdot]V^\dag =\mathcal{\tilde{L}}_t[V\cdot V^\dag]$ for all $t$, respectively.
It is far from trivial how these weak symmetries affect dynamics of open quantum systems.
Furthermore, with such distinctions of (strong or weak) symmetries between Schr\"odinger and interaction pictures, the notion of the spontaneous symmetry breaking and phase transition in open quantum many-body systems~\cite{minganti_SpectralTheoryLiouvillians_2018,lieu_SymmetryBreakingError_2020,lee_QuantumCriticalityDecoherence_2023,sala_SpontaneousStrongSymmetry_2024,gu_SpontaneousSymmetryBreaking_2025,ziereis_StrongtoWeakSymmetryBreaking_2025,lessa_StrongtoWeakSpontaneousSymmetry_2025} could be enriched, which merits further studies.

Fourth, our work opens up the study of open quantum systems recurrently driven in time.
As we have already discussed with concrete examples of time-quasiperiodic driving (including many-body systems), such systems behave differently from time-independent and time-periodic open quantum dynamics.
Given the fact that time-periodic open quantum dynamics~\cite{prosen_NonequilibriumPhaseTransition_2011,vorberg_GeneralizedBoseEinsteinCondensation_2013,haddadfarshi_CompletelyPositiveApproximate_2015,dai_FloquetTheoremOpen_2016,fitzpatrick_ObservationDissipativePhase_2017,hartmann_AsymptoticFloquetStates_2017,ikeda_GeneralDescriptionNonequilibrium_2020,schnell_ThereFloquetLindbladian_2020,fedorov_PhotonTransportBoseHubbard_2021,schnell_HighfrequencyExpansionsTimeperiodic_2021,mizuta_BreakdownMarkovianityInteractions_2021,ikeda_NonequilibriumSteadyStates_2021,mori_FloquetStatesOpen_2023,chen_PeriodicallyDrivenOpen_2024} and time-quasiperiodic closed quantum dynamics~\cite{ho_SemiclassicalManymodeFloquet_1983,casati_AndersonTransitionOneDimensional_1989,martin_TopologicalFrequencyConversion_2017,malz_TopologicalTwodimensionalFloquet_2021,beatrez_ObservationCriticalPrethermal_2023,park_ControllableFloquetEdge_2023,roy_SingleMultifrequencyDriving_2024,he_ExperimentalRealizationDiscrete_2025,maity_FibonacciSteadyStates_2019,lapierre_FineStructureHeating_2020,pilatowsky-cameo_CompleteHilbertSpaceErgodicity_2023} are intensively investigated to engineer quantum states, recurrently (e.g., quasiperiodically) driven open quantum systems could provide another useful route to control quantum dynamics, for which our results will serve as a rigorous foundation.

Last but not least, beyond the context of open quantum systems, our theory may provide new insights into ergodic theory in mathematical physics.
Ergodic theory constitutes a fundamental framework underlying classical and quantum stochastic processes, the theory of matrix product states~\cite{fannes_FinitelyCorrelatedStates_1992b}, and statistical physics~\cite{deutsch_QuantumStatisticalMechanics_1991}.
Its relevance extends to a wide range of active research topics, including measurement-induced phase transitions~\cite{hamazaki_IntroductionMonitoredQuantum_2026} and deep thermalization~\cite{mark_MaximumEntropyPrinciple_2024}. The perspective developed in this work, such as the recurrent functions and symmetry classifications, may contribute more broadly to understanding how ergodic and non-ergodic behavior emerge in quantum dynamics.

\medskip

{\it Note added.}-- After completion of our work, we became aware of Ref.~\cite{anisur_DissipativeDickeTime_2026}, which discusses the emergence of time quasicrystals in an open Dicke model subjected to a quasiperiodic Fibonacci drive.

\begin{acknowledgments}
We thank Hal Tasaki for insightful comments. This work was supported by JST ERATO Grant No. JPMJER2302, Japan.
H.~Y.\ acknowledges support by the Special Postdoctoral Researchers Program at RIKEN. R.H. was supported by JSPS KAKENHI Grant No. JP24K16982.
\end{acknowledgments}

\appendix

\section{Proof that Example~\ref{ex:multi_frequency_drives} satisfies Condition~\ref{con:recurrent}}
\label{sec:multi_freq_proof}
Here, we prove that Example~\ref{ex:multi_frequency_drives} satisfies Condition~\ref{con:recurrent}.
The Liouvillian $\mathcal{L}_t$ for Example~\ref{ex:multi_frequency_drives} is given in the form 
$\mathcal{L}_t=\sum_{k=1}^K e^{i\lambda_k t}\mathcal{L}^{(k)}$,
where $K\in \mathbb{N}$ and $\lambda_k\in \mathbb{R}$ are constants and $\mathcal{L}^{(k)}$ are time-independent superoperators. 
Thus, $\mathcal{L}_t$ is continuous and bounded because
\begin{equation}
    \opn\mathcal{L}_t \opn_2 \leq \sum_{k=1}^K \opn\mathcal{L}^{(k)}\opn_2=:M
\end{equation}
by the triangle inequality. Furthermore, it is known that $\mathcal{L}_t$ is almost periodic, i.e., for every $\varepsilon>0$ there exists $l=l_\varepsilon$ such that for every $a\in\mathbb{R}$, there is a number $a<\tau<a+l$ with the property
\begin{equation}
    \opn\mathcal{L}_{\tau+t} - \mathcal{L}_{t} \opn_2<\varepsilon,\quad \forall t\in\mathbb{R}
\end{equation}
({Proposition 3.20} of \cite{corduneanu_AlmostPeriodicOscillations_2009}). Such numbers are called $\varepsilon$-translation numbers.
Thus, for any $\delta>0$, we can take the infinite sequence of $\varepsilon$-translation numbers $\{\tau_n\}_{n=1,2,\ldots}$
such that $\tau_{n+1}-\tau_n>\delta$. For any $t_0\geq0$, by defining $t_n=t_0+\tau_n$, we have
\begin{equation}
    \opn\mathcal{L}_{t_n+t} - \mathcal{L}_{t_0+t} \opn_2=\opn\mathcal{L}_{\tau_n+(t_0+t)} - \mathcal{L}_{(t_0+t)} \opn_2<\varepsilon
\end{equation}
for all $t\in [0,\delta]$. 

\section{Proof of Lemma~\ref{lem:lambda0_bound}}
\label{sec:main1_lemma1}
We prove Lemma~\ref{lem:lambda0_bound} by contrapositive. The contrapositive reads:
``If $\lambda_0=1$ for any $t_0\geq0$ and $\delta>0$, then $\cint\neq\{c\mathbb{I}\mid c\in \mathbb{C}\}$''

Assume $\lambda_0=1$ for any $t_0\geq0$ and $\delta>0$. Then, there exists a non-zero and traceless solution $\pi_t$ of the GKSL equation~\eqref{eq:GKSL} such that 
\begin{equation}
   \|\pi_t\|_2=\|\pi_0\|_2>0 \quad\forall t\geq0. 
\end{equation}
Next, we move to the interaction picture by~\eqref{eq:transform_int}. Note that the Hilbert-Schmidt norm is invariant under this, because $\|\tilde{\pi}_t\|_2=\|U_t^\dagger\pi_tU_t\|_2=\|\pi_t\|_2$. Thus, $\tilde{\pi}_t$ is a solution of the GKSL equation in the interaction picture~\eqref{eq:GKSL_int} such that 
\begin{equation}
    \|\tilde{\pi}_t\|_2=\|\tilde{\pi}_0\|_2>0 \quad\forall t\geq0
\end{equation}
Using the GKSL equation~\eqref{eq:GKSL_int}, we have
\begin{equation}
    \partial_t\|\tilde{\pi}_t\|_2=-\frac{\sum_{m=1}^M (\|[\tilde{L}_{m,t},\tilde{\pi}_t]\|_2)^2}{2\|\tilde{\pi}_t\|_2}=0 \quad\forall t\geq0.
\end{equation}
Then, for all $m$,
\begin{equation}
   [\tilde{L}_{m,t},\tilde{\pi}_t]=0\quad \forall t\geq0.
   \label{eq:proof_strong_int}
\end{equation}
Thus, again by the GKSL equation~\eqref{eq:GKSL_int},
\begin{align}
  \frac{d \tilde{\pi}_t}{d t}
  =-\frac{1}{2}\sum_{m}
  [\tilde{L}_{m,t},[\tilde{L}_{m,t},\tilde{\pi}_t]]=0 \quad \forall t\geq0.
\end{align}
Therefore, 
\begin{equation}
    \tilde{\pi}_t=\tilde{\pi}_0=\text{const.} \quad\forall t\geq0.
\end{equation}
From this and \eqref{eq:proof_strong_int}, one finds $\tilde{\pi}_0\in\cint$. Since $\tilde{\pi}_0$ is a nonzero and traceless operator, $\cint\neq\{c\mathbb{I}\mid c\in \mathbb{C}\}$. This proves the contrapositive, and hence Lemma~\ref{lem:lambda0_bound} holds.

\section{Proof of Lemma~\ref{lem:lambdan_bound}}
\label{sec:main1_lemma2}
{Condition~\ref{con:recurrent}} states that for all $\varepsilon>0$, there exists an infinite sequence of mutually disjoint time intervals $[t_n,t_n+\delta]$ $(n=0,1,\ldots)$ such that
\begin{equation}
    \opn\mathcal{L}_{t_n+t} - \mathcal{L}_{t_0+t} \opn_2<\varepsilon \quad \forall t\in [0,\delta].
\end{equation}
From~\eqref{eq:inequality_prime}, we have 
\begin{equation}
    \opn\mathcal{L}_{t_n+t} - \mathcal{L}_{t_0+t} \opn^\prime_2<\varepsilon \quad \forall t\in [0,\delta].
    \label{eq:bound_prime_l}
\end{equation}
From the triangle inequality, we have
\begin{align}
    |\lambda_n-\lambda_0|
    &\leq\opn\mathcal{V}_{t_n,t_n+\delta}-\mathcal{V}_{t_0,t_0+\delta}\opn_2^\prime.
    \label{eq:bound_lambda_v}
\end{align}
By writing $\mathcal{D}_t=\mathcal{V}_{t_n,t_n+t}-\mathcal{V}_{t_0,t_0+t}$, it holds that 
\begin{equation}  
    \partial_t\mathcal{D}_t=\mathcal{L}_{t_n+t}\mathcal{D}_t+(\mathcal{L}_{t_n+t} - \mathcal{L}_{t_0+t})\mathcal{V}_{t_0,t_0+t}.
\end{equation}
From \eqref{eq:bound_phi1phi2} and the triangle inequality, we have
\begin{equation}  
    \partial_t \opn\mathcal{D}_t\opn_2^\prime\leq\opn\mathcal{L}_{t_n+t}\opn_2^\prime\opn\mathcal{D}_t\opn_2^\prime+\opn\mathcal{L}_{t_n+t} - \mathcal{L}_{t_0+t}\opn_2^\prime\opn\mathcal{V}_{t_0,t_0+t}\opn_2^\prime.
\end{equation}
Here, we used that $\partial_t\opn\mathcal{D}_t\opn_2^\prime\leq \opn\partial_t\mathcal{D}_t\opn_2^\prime$, which follows from the triangle inequality.
    
From~{Condition~\ref{con:recurrent}}, $\opn\mathcal{L}_t\opn^\prime_2\leq M$ for all $t$. This, together with
\begin{equation}
    \opn\mathcal{V}_{t_0,t_0+t}\opn_2^\prime\leq  e^{\int_0^t \opn\mathcal{L}_s\opn_2^\prime ds } \leq e^{Mt}
\end{equation}
for any $0\leq t \leq \delta$ leads to
\begin{equation}
    \partial_t \opn\mathcal{D}_t\opn_2^\prime\leq M\opn\mathcal{D}_t\opn_2^\prime+e^{Mt}\opn\mathcal{L}_{t_n+t} - \mathcal{L}_{t_0+t}\opn_2^\prime.
\end{equation}
for any $0\leq t \leq \delta$. By applying Gr\"onwall's inequality,
\begin{equation}
    \opn\mathcal{D}_\delta\opn_2^\prime\leq e^{M\delta}\int_0^\delta\opn\mathcal{L}_{t_n+t} - \mathcal{L}_{t_0+t}\opn_2^\prime dt.
\end{equation}
Here, we used $D_0=\mathcal{V}_{t_n,t_n}-\mathcal{V}_{t_0,t_0}=0$. From \eqref{eq:bound_prime_l}, we have
\begin{equation}
    \opn\mathcal{D}_\delta\opn_2^\prime<\varepsilon\delta e^{M\delta}.
\end{equation}
This, together with \eqref{eq:bound_lambda_v}, yields
\begin{equation}
   |\lambda_n-\lambda_0|<\varepsilon\delta e^{M\delta}.
\end{equation}

\section{Proof of Theorem~\ref{thm:main2}}\label{sec:proof_main2}
From Theorem~\ref{thm:main1}, the steady state is unique if and only if $\cint=\{c\mathbb{I}\mid c\in\mathbb{C}\}$. Thus, it suffices to prove that 
\begin{equation}
\mathcal{C}_{t_0}^{\mathrm{ad}} =\{c\mathbb{I}\mid c\in\mathbb{C}\} \Longrightarrow
  \cint=\{c\mathbb{I}\mid c\in\mathbb{C}\}. 
  \label{eq:proof_main2_1}
\end{equation}
if $H_{t}$ and $L_{m,t}$ are smooth at $t=t_0$ and 
\begin{equation}
\cint=\{c\mathbb{I}\mid c\in\mathbb{C}\} \Longrightarrow\mathcal{C}_{t_0}^{\mathrm{ad}} =\{c\mathbb{I}\mid c\in\mathbb{C}\}
\label{eq:proof_main2_2}
\end{equation}
for some $t_0$ if $H_{t}$ and $L_{m,t}$ are analytical for every $t$.

First, from \eqref{eq:unitary_cad}, it holds that 
\begin{equation}
  \mathcal{C}_{t_0}^{\mathrm{ad}} =\{c\mathbb{I}\mid c\in\mathbb{C}\}\iff\tilde{\mathcal{C}}_{t_0}^{\mathrm{ad}} =\{c\mathbb{I}\mid c\in\mathbb{C}\}. 
\end{equation}
Next, we prove
\begin{equation}
    \tilde{\mathcal{C}}_{t_0}^{\mathrm{ad}}\supseteq\cint
    \label{eq:proof_main2_3}
\end{equation}    
if $H_{t}$ and $L_{m,t}$ are smooth at $t=t_0$ and 
\begin{equation}
    \tilde{\mathcal{C}}_{t_0}^{\mathrm{ad}}\subseteq\cint
    \label{eq:proof_main2_4}
\end{equation}  
if $H_{t}$ and $L_{m,t}$ are analytical for every $t$, which leads to \eqref{eq:proof_main2_1} and \eqref{eq:proof_main2_2}, respectively.

First, we prove~\eqref{eq:proof_main2_3} by proving that if $A\in\cint$, then $A\in\tilde{\mathcal{C}}_{t_0}^{\mathrm{ad}}$. From $A\in\cint$, 
\begin{equation}
    [\tilde{L}_{m,t},A]=0\quad\forall{m,t}.
    \label{eq:proof_main2_5}
\end{equation}
If $H_{t}$ and $L_{m,t}$ are smooth at $t=t_0$, so is $\tilde{L}_{m,t}$. Thus, by taking the $n$-th derivative of both sides of \eqref{eq:proof_main2_5}, we have
\begin{equation}
    [\partial^n_t\tilde{L}_{m,t}|_{t=t_0},A]=0\quad\forall{m,n},
\end{equation}
which leads to  $A\in\tilde{\mathcal{C}}_{t_0}^{\mathrm{ad}}$. Therefore, \eqref{eq:proof_main2_3} holds.

Next, we prove~\eqref{eq:proof_main2_4} by proving that if $A\in\tilde{\mathcal{C}}_{t_0}^{\mathrm{ad}}$ for some $t_0$, then $A\in\cint$. If $H_{t}$ and $L_{m,t}$ are analytical for every $t$, so is $\tilde{L}_{m,t}$, i.e., for every $t_0$, $\tilde{L}_{m,t_0}$ admits a Taylor expansion~\eqref{eq:taylor} in some open neighborhood $U(t_0)$ of $t=t_0$. Thus, assuming that $A\in\tilde{\mathcal{C}}_{t_0}^{\mathrm{ad}}$ for some $t_0$, i.e.,
\begin{equation}
   [\partial_t^n \tilde{L}_{m,t} \big|_{t=t_0},A] = 0 \quad\forall{m,n}
\end{equation}
for some $t_0$, then every terms in the right-hand side of~\eqref{eq:taylor} commutes with $A$, thus
\begin{equation}
[\tilde{L}_{m,t},A] = 0 \quad\forall{m}\text{ and }t\in U(t_0).
\end{equation}

Thus we can take $t_1\neq t_0$ such that $U(t_0) \cap U(t_1)\neq \emptyset$ and 
\begin{equation}
[\partial_t^n \tilde{L}_{m,t} \big|_{t=t_1},A] = 0 \quad\forall{m,n},
\end{equation}
because
\begin{equation}
[\tilde{L}_{m,t},A] = 0 \quad\forall{m}\text{ and }t\in U(t_0) \cap U(t_1).
\end{equation}
Therefore,
\begin{equation}
[\tilde{L}_{m,t},A] = 0 \quad\forall{m}\text{ and }t\in U(t_0) \cup U(t_1).
\end{equation}
Since $[0,\infty)$ is connected, we can repeat this argument to prove that
\begin{equation}
[\tilde{L}_{m,t},A] = 0 \quad\forall{m}\text{ and }t\in [0,\infty).
\end{equation}
Therefore, $A\in\cint$, which proves~\eqref{eq:proof_main2_4}.

\section{Counterexample to the converse part of Theorem~\ref{thm:main2} for a non-analytical GKSL equation}
\label{sec:counterexample}

Here, we provide a counterexample to the converse part of Theorem~\ref{thm:main2} for a non-analytical GKSL equation, which illustrates the importance of the assumption of analyticity of $\{H_t, \{L_{m,t}\}\}$.

\begin{ex}[Non-analytical case]
\label{ex:bump_function}

Consider a periodic Liouvillian $\mathcal{L}_t$ that is analytical at not all $t$ defined by:
\begin{align}
    H_t&=\begin{cases}
        g\sigma^x & \text{if } (n+\frac{1}{2})T\leq t\leq (n+\frac{3}{4})T \ (n=0,1,\ldots),\\
        0& \text{otherwise, } \\
    \end{cases}\\
    L_t&=\begin{cases}
        \sqrt{\kappa}\sigma^z& \text{if }nT\leq t\leq (n+\frac{1}{4})T\ (n=0,1,\ldots),\\
        0& \text{otherwise, } \\
    \end{cases}
\end{align}
where $g,\kappa>0$ are constants and $\sigma^\mu$ $(\mu=x,y,z)$ are the Pauli operators. This Liouvillian is not continuous in $t$, but we can make it continuous by coarse-graining (See Eq.~\eqref{eq:coarse‑graining}). Still, the coarse-grained Liouvillian is not analytical. 

This example can serve as a touchstone for assessing the validity of Theorem~\ref{thm:main2} for non-analytical GKSL equations by considering the one-cycle time-evolution superoperator
\begin{align}
\mathcal{U}_F = \mathcal{V}_{0,T},
\label{eq:time-evolution}
\end{align}
where $\mathcal{V}_{s,t}= \mathcal{T}e^{\int_{s}^{t}\mathcal{L}_{t'}dt'}$. 
Suppose that $\mathcal{U}_F(\rho^*)=z\rho^*$ with $|z|=1$.
Then
\begin{equation}
\rho^*_t = \mathcal{V}_{0,t}(\rho^*)
\end{equation}
defines a steady state that is generally time dependent.
Note that $\mathbb{I}/d$ always provides a time-independent steady state. The uniqueness of the steady state reduces to the uniqueness of eigenvalues of $\mathcal{U}_F$ with modulus $1$. If  such an eigenvalue is unique, $\mathcal{U}_F$ is called
\emph{mixing}~\cite{wolf_QuantumChannelsOperations_2012,burgarth_ErgodicMixingQuantum_2013}.
Using the Kraus representation, one can determine whether $\mathcal{U}_F$ is mixing. 

By integrating the equation, one obtains a Kraus representation of $\mathcal{U}_F$ as
\begin{equation}
    \mathcal{U}_F(\rho)=\mathcal{V}_{0,T}(\rho)= \sum_{j=0,1} K_j \rho K_j^\dagger,
\end{equation}
where
\begin{align}
    K_0&= \sqrt{\frac{1+e^{-\kappa T/2}}{2}} \left[\cos\left(\frac{gT}{4}\right)\mathbb{I}-i\sin\left(\frac{gT}{4}\right)\sigma^x\right],\\
    K_1&= \sqrt{\frac{1-e^{-\kappa T/2}}{2}}\left[\cos\left(\frac{gT}{4}\right)\sigma^z-\sin\left(\frac{gT}{4}\right)\sigma^y\right].
\end{align}

The structure of the steady state depends on $gT$.

\begin{enumerate}
    \item When $\sin(gT/4)=0$, i.e., $gT=4m\pi$ for some integer $m$, $K_0\propto\mathbb{I}$ and  $K_1\propto \sigma^z$. Thus, $\mathcal{U}_F$ has two eigenmodes with eigenvalue 1: $\mathbb{I}$ and $\sigma^z$. Therefore, the steady state is not unique, and the eigenmode $\sigma^z$ leads to a time-dependent steady state.
    
    \item When $\cos(gT/4)=0$, i.e., $gT=(4m-2)\pi$ for some integer $m$, $K_0\propto\sigma^x$ and  $K_1\propto \sigma^y$. Thus, $\sigma^z$ is an eigenmode of $\mathcal{U}_F$ with eigenvalue $-1$. As a result, the steady state is not unique, with the eigenmode $\sigma^z$ leading to a time-dependent steady state.

    \item Finally, when $\sin(gT/4),\,\cos(gT/4)\neq0$, i.e., $gT\neq 2m\pi$ for any integer $m$, $\mathcal{U}_F$ is mixing. This can be proved by checking that 
    $\operatorname{span}\{K_0^2,K_0K_1, K_1K_0, K_1^2\}=\mathcal{B}(\mathcal{H})$~\cite{wolf_QuantumChannelsOperations_2012,burgarth_ErgodicMixingQuantum_2013}. By writing $s=\sin(gT/2)$ and $c=\cos(gT/2)$, we have
    \begin{align}
        K_0^2&\propto c\mathbb{I}-is\sigma^x,\\
        K_0K_1&\propto c\sigma^z-s\sigma^y,\\
        K_1K_0&\propto \sigma^z,\\
        K_1^2&\propto \mathbb{I}.
    \end{align}
    Since $gT\neq 2m\pi$ for any integer $m$, $s\neq0$ and $c\neq1,-1$. Thus,
    \begin{align}
        \mathbb{I}&\propto K_1^2,\\
        \sigma^x&\propto K_0^2-cK_1^2,\\
        \sigma^y&\propto K_0K_1-cK_1 K_0,\\
        \sigma^z&\propto K_1 K_0,
    \end{align}
    which imply that $\operatorname{span}\{K_0^2,K_0K_1, K_1K_0, K_1^2\}=\mathcal{B}(\mathcal{H})$. Therefore, $\mathcal{U}_F$ is mixing, which means that the GKSL equation has a unique steady state. 
\end{enumerate}

Thus, if $gT\neq 2m\pi$ for any integer $m$, the GKSL equation has a unique steady state. However, for any $gT$, there is no $t_0$ such that $\mathcal{C}^\mathrm{ad}_{t_0}=\{c\mathbb{I}\mid c\in\,\mathbb{C}\}$, thus the converse part of Theorem~\ref{thm:main2} does not hold. 
This is because the inclusion relation~\eqref{eq:proof_main2_4} does not hold. To see this, we calculate $\cint$. First, $\tilde{L}_{t}=U_t^\dagger L_tU_t$ is given by:
\begin{equation}
    \tilde{L}_{t}=\begin{cases}
        \sqrt{\kappa}(\sigma^z\cos\theta_n+\sigma^y\sin\theta_n)& \text{if } nT\leq t\leq (n+\frac{1}{4})T,\\
        0& \text{otherwise}, \\
    \end{cases}
\end{equation}
where $\theta_n=\frac{ngT}{2}$. Thus, $\tilde{L}_{t}$ is non-analytical at $t=nT, (n+\frac{1}{4})T$. Furthermore, $\cint$ depends on $gT$, because if $gT=2m\pi$ for some integer $m$, then $\theta_n=m\pi$, so $\tilde{L}_{t}\propto\sigma^z$ for all $t$. Therefore,
\begin{equation}
    \cint=\begin{cases}
        \{c_1\mathbb{I}+c_2\sigma^z\mid c_1,c_2\in\,\mathbb{C}\}& \text{if }gT=2m\pi, \\
        \{c\mathbb{I}\mid c\in\,\mathbb{C}\}& \text{otherwise}. \\
    \end{cases}
\end{equation}
Now, even when $\cint=\{c\mathbb{I}\mid c\in\,\mathbb{C}\}$, there is no $t_0$ such that $\mathcal{C}^\mathrm{ad}_{t_0}=\{c\mathbb{I}\mid c\in\,\mathbb{C}\}$, which means that inclusion relation~\eqref{eq:proof_main2_4} does not hold. Note that this example is still consistent with Theorem~\ref{thm:main1}, because the steady state is unique if and only if $gT\neq2m\pi$ from the discussion using the Kraus representation, which exactly corresponds to the condition for $\cint=\{c\mathbb{I}\mid c\in\,\mathbb{C}\}$.
We also note that, for $gT=2m\pi$, we have $\cint\neq\{c\mathbb{I}\mid c\in\,\mathbb{C}\}$ with $\cst=\{c\mathbb{I}\mid c\in\,\mathbb{C}\}$, which correctly predict a time-dependent steady state without non-trivial time-independent steady states.

\end{ex}

\section{Proof of Lemma~\ref{lem:tfim2}}
\label{sec:proof_lem7}

To prove Lemma~\ref{lem:tfim2}, we first establish the following lemma.

\begin{lem}
Consider the GKSL equation defined through Eq.~\eqref{eq:example_3_majorana}.
For all integers $n$ with $0\leq n\leq V-1$, it holds that
\label{lem:tfim}
    \begin{align}
    \ad_t^{n}(L_t)&=f_{n,t}\gamma_{n+1}+g_{n,t}(\gamma_1,\ldots,\gamma_{n}),\label{eq:adn}
\end{align}
where 
\begin{align}  f_{2m,t}&=2^{2m}\sqrt{\kappa}B^m\cos^m(\omega t),\label{eq:feven}\\
f_{2m+1,t}&=2^{2m+1}\sqrt{\kappa}B^{m+1}\cos^{m+1}(\omega t)\label{eq:fodd}
\end{align}
and $g_{n,t}(\gamma_1,\ldots,\gamma_{n})$ is a linear combination of $\gamma_1,\ldots,\gamma_{n}$ with coefficients that depend smoothly on $t$. 
\end{lem}

\begin{proof}[Proof of Lemma~\ref{lem:tfim}]
We prove Lemma~\ref{lem:tfim} by induction on $n$. For the base case $n=0$ and $n=1$, 
\begin{align}
 \ad_t^{0}(L_t)&=\sqrt{\kappa}\gamma_1,\\
 \ad_t^{1}(L_t)&=2\sqrt{\kappa}B\cos(\omega t)\gamma_2,
\end{align}
which are consistent with Eqs.~\eqref{eq:adn}, \eqref{eq:feven}, and \eqref{eq:fodd}. Now, assume {Lemma~\ref{lem:tfim}} holds for $n=2k-1$ with $1\leq k \leq V-1$. Then, for $n=2k$, we compute:
\begin{align}
\ad_t^{2k}(L_t)&=i[H_t,\ad_t^{2k-1}(L_t)]+\partial_t[\ad_t^{2k-1}(L_t)]\nonumber\\
&=-f_{2k-1,t}[B\cos(\omega t)\gamma_{2k-1}\gamma_{2k}+\gamma_{2k}\gamma_{2k+1},\gamma_{2k}]\nonumber\\
&\ \ \ \ +\partial_tf_{2k-1,t}\gamma_{2k}+\partial_tg_{2k-1,t}(\gamma_1,\ldots,\gamma_{2k-1})\nonumber\\
&=f_{2k,t}\gamma_{2k+1}+g_{2k,t}(\gamma_1,\ldots,\gamma_{2k})
\end{align}
where
\begin{align}
f_{2k,t}&=2f_{2k-1,t},\\
g_{2k,t}(\gamma_1,\ldots,\gamma_{2k})
&=-2B\cos(\omega t)f_{2k-1,t}\gamma_{2k-1}\nonumber\\
&\ \ \ \ +\partial_tf_{2k-1,t}\gamma_{2k}\nonumber\\
&\ \ \ \ +\partial_tg_{2k-1,t}(\gamma_1,\ldots,\gamma_{2k-1}).
\end{align}
Next, for $n=2k+1$, we have:
\begin{align}
\ad_t^{2k+1}(L_t)&=i[H_t,\ad_t^{2k}(L_t)]+\partial_t[\ad_t^{2k}(L_t)]\nonumber\\
&=-f_{2k,t}[\gamma_{2k}\gamma_{2k+1}+B\cos(\omega t)\gamma_{2k+1}\gamma_{2k+2},\gamma_{2k+1}]\nonumber\\
&\ \ \ \ +\partial_tf_{2k,t}\gamma_{2k+1}+\partial_tg_{2k,t}(\gamma_1,\ldots,\gamma_{2k})\nonumber\\
&=f_{2k+1,t}\gamma_{2k+2}+g_{2k+1,t}(\gamma_1,\ldots,\gamma_{2k+1}),
\end{align}
where
\begin{align}
f_{2k+1,t}&=2B\cos(\omega t)f_{2k,t},\\
g_{2k+1,t}(\gamma_1,\ldots,\gamma_{2k+1})
&=-2f_{2k,t}\gamma_{2k}\nonumber\\
&\ \ \ \ +\partial_tf_{2k,t}\gamma_{2k+1}\nonumber\\
&\ \ \ \ +\partial_tg_{2k,t}(\gamma_1,\ldots,\gamma_{2k}).
\end{align}
Thus, {Lemma~\ref{lem:tfim}} also holds for $n=2k$ and $n=2k+1$. By the principle of mathematical induction, we conclude that {Lemma~\ref{lem:tfim}} holds for all $n=0,\ldots, 2V-1$.
\end{proof}

\begin{proof}[Proof of Lemma~\ref{lem:tfim2}]
We prove Lemma~\ref{lem:tfim2} by induction on $n$. For the base case $n=0$, we have
\begin{align}
    L_t|_{t=0}&\propto\gamma_1.
\end{align}
Assume that Lemma~\ref{lem:tfim2} holds for $n=k-1$ with $1\leq k\leq V-1$. From Lemma~\ref{lem:tfim}, for $n=k$, we obtain
\begin{align}
    \ad_t^{k}(L_t)|_{t=0}&=f_{k,0}\gamma_{k+1}+g_{k,0}(\gamma_1,\ldots,\gamma_{k}),
\end{align}
which can be rearranged as
\begin{align}
\gamma_{k+1}=f_{k,0}^{-1}\left[\ad_t^{k}(L_t)|_{t=0}-g_{k,0}(\gamma_1,\ldots,\gamma_{k})\right],
    \label{eq:sigmaz}
\end{align}
where we used that $f_{k,0} \neq 0$ for all $k$. By the induction hypothesis, the right-hand side of equation~\eqref{eq:sigmaz} can be generated from the set $\{\mathbb{I},\ad_t^{m}(L_t)|_{t=0}\mid m=0,1,\ldots,k\}$.
Therefore, Lemma~\ref{lem:tfim2} also holds for $n=k$. By the principle of mathematical induction, we conclude that Lemma~\ref{lem:tfim2} holds for all $n=0,\ldots, 2V-1$.
\end{proof}

\section{Proof of Theorem~\ref{thm:indep}}
\label{sec:proof_thm_indep}
(\emph{1}. $\to$ \emph{2}.) 
We proceed by contradiction. Assume that
\begin{equation}
    \lim_{t\to\infty} \rho_t = \rho^*,
    \label{eq:proof_thm8_1}
\end{equation}
for some initial condition $\rho_0$, but 
\begin{equation}
\rho^* \notin \cst.
\label{eq:proof_thm8_2}
\end{equation}
From~\eqref{eq:proof_thm8_2}, equation \eqref{eq:commutant_degen} implies that
\begin{equation}
    \mathcal{L}_{t_0} (\rho^*)\neq0 
\end{equation}
for some $t_0\geq0$. Thus, we can choose a non-zero element of the matrix $\mathcal{L}_{t_0} (\rho^*)$, and we write it as $[\mathcal{L}_{t_0} (\rho^*)]_{kl}$. Without loss of generality, we assume that 
\begin{equation}
f:=\operatorname{Re}[\mathcal{L}_{t_0} (\rho^*)]_{kl}>0.
\end{equation}
Since $\mathcal{L}_t$ is continuous in time, there exists $\delta>0$ such that
\begin{equation}
    \operatorname{Re}[\mathcal{L}_{t} (\rho^*)]_{kl}>\frac{f}{2}>0 \quad \forall t\in[t_0,t_0+\delta].
\end{equation}
From Condition~\ref{con:recurrent}, for any $\varepsilon>0$, we can take an infinite sequence of mutually disjoint time intervals $[t_n,t_n+\delta]$ $(n=0,1,\ldots)$ such that
\begin{equation}
      \operatorname{Re}[\mathcal{L}_{t} (\rho^*)]_{kl}>\frac{f}{2}-\varepsilon \quad \forall t\in[t_n,t_n+\delta].
      \label{eq:proof_thm8_3}
\end{equation}

Next, from~\eqref{eq:proof_thm8_1}, for any $\alpha$, there exists $\tilde{t}$ such that
\begin{equation}
    \|\rho_t-\rho^*\|_2<\alpha\quad \forall t>\tilde{t}.
    \label{eq:proof_thm8_4}
\end{equation}
From Condition~\ref{con:recurrent}, $\opn\mathcal{L}_t\opn_2\leq M$, and thus
\begin{equation}
    \|\mathcal{L}_t(\rho_t)-\mathcal{L}_t(\rho^*)\|_2<M\alpha\quad \forall t>\tilde{t}.
    \label{eq:proof_thm8_5}
\end{equation}
By taking $n$ such that $\tilde{t}<t_n$, from \eqref{eq:proof_thm8_3} and \eqref{eq:proof_thm8_5}, we have
\begin{equation}
    \operatorname{Re}[\mathcal{L}_{t} (\rho_t)]_{kl}>\frac{f}{2}-M\alpha-\varepsilon \quad \forall t\in[t_n,t_n+\delta].
\end{equation}
Therefore, it holds that
\begin{align}
    \| \rho_{t_n+\delta}-\rho_{t_n}\|_2
    &\geq |[\rho_{t_n+\delta}-\rho_{t_n}]_{kl}|\nonumber\\
    &\geq \operatorname{Re}[\rho_{t_n+\delta}-\rho_{t_n}]_{kl}\nonumber\\
    &=\int_{t_n}^{t_n+\delta} \operatorname{Re} [\mathcal{L}_t(\rho_t)]_{kl}dt\nonumber\\
    &>\left(\frac{f}{2}-M\alpha-\varepsilon\right)\delta.
    \label{eq:proof_thm8_6}
\end{align}
On the other hand, equation \eqref{eq:proof_thm8_4} implies that
\begin{equation}
    \|\rho_{t_n+\delta}-\rho_{t_n}\|<2\alpha.
    \label{eq:proof_thm8_7}
\end{equation}
By comparing~\eqref{eq:proof_thm8_6} and \eqref{eq:proof_thm8_7}, it must hold that
\begin{equation}
   \left(\frac{f}{2}-M\alpha-\varepsilon\right)\delta <2\alpha.
\end{equation}
Since $\alpha$ and $\varepsilon$ can be taken to be arbitrarily small, by setting $\alpha<\min (\frac{f}{8M},\frac{f\delta}{8})$ and $\varepsilon<\frac{f}{8}$, we have
\begin{equation}
   \left(\frac{f}{2}-M\alpha-\varepsilon\right)\delta>\frac{f\delta}{4}>2\alpha.
\end{equation}
However, this is a contradiction. Therefore, the assumption that $\rho^* \notin \cst$ must be incorrect.  
Hence, if $\lim_{t\to\infty} \rho_t = \rho^*$ for some initial condition $\rho_0$, then $\rho^* \in \cst$.

(\emph{2}. $\to$ \emph{1}.) If $\rho^* \in \cst$, then $ \mathcal{L}_t (\rho^*) = 0$ for all $t$. Thus, by choosing $\rho_0=\rho^*$, we have $\lim_{t\to \infty} \rho_t = \rho^*$.

\section{Proof of Theorem~\ref{thm:time_dependent}}
\label{sec:proof_thm_time_dependent}

(\emph{1}. $\to$ \emph{2}.) Suppose $\cint\setminus\cst\neq \emptyset$. Since $\cint$ and $\cst$ are closed under Hermitian conjugation and they include $\mathbb{I}$, there is a density matrix $\rho^*$ such that 
\begin{equation}
    \rho^* \in\cint \text{ and } \rho^* \notin\cst.
\end{equation}
Since $\rho^* \in\cint$, $\tilde{\mathcal{L}}_t(\rho^*)=0$ for all $t$, and 
$\tilde{\rho}_t=\rho^*$ is a solution of the GKSL equation in the interaction picture~\eqref{eq:GKSL_int}. Thus, by relation~\eqref{eq:transform_int}, $\rho_t=U_t\rho^*U_t^\dagger$ is a solution in the Schr\"odinger picture.

We prove that $\rho_t=U_t\rho^*U_t^\dagger$ does not converge to any time-independent density matrix as $t\to \infty$ by contradiction. Assume that 
\begin{equation}
    \lim_{t\to\infty}U_t\rho^*U_t^\dagger=\sigma^*
    \label{eq:time_dependent_proof1}
\end{equation}
for some time-independent density matrix $\sigma^*$. Then, by {Theorem~\ref{thm:indep}}, $\sigma^*\in\cst$, and therefore $U_t^\dagger\sigma^* U_t=\sigma^*$ for all $t$. In particular,  \begin{equation}
   \lim_{t\to\infty}U_t^\dagger\sigma^*U_t=\sigma^*
   \label{eq:time_dependent_proof2}
\end{equation}
From \eqref{eq:time_dependent_proof1}, it holds that
\begin{equation}
    \lim_{t\to\infty}U_t^\dagger\sigma^*U_t=\rho^*
    \label{eq:time_dependent_proof3}
\end{equation}
From \eqref{eq:time_dependent_proof2} and \eqref{eq:time_dependent_proof3}, we find $\sigma^*=\rho^*$. However, since $\rho^* \notin\cst$, we also have $\sigma^*\neq\rho^*$, and this is a contradiction. Therefore, $\rho_t=U_t\rho^*U_t^\dagger$ does not converge to any time-independent density matrix as $t\to \infty$, which means that $\rho_t$ is a time-dependent steady state.

(\emph{2}. $\to$ \emph{1}.) Assuming $\cint=\cst$, we prove that any initial state relaxes to a time-independent state. From Lemma~\ref{lem:cint}, if
\begin{equation}
    \lim_{t\to \infty} \rho_t = \rho^*_t,
\end{equation}
there exists some density matrix $\tilde{\rho}^*\in \cint$ and 
\begin{equation}
   \rho^*_t=U_t \tilde{\rho}^*U_t^\dagger.
\end{equation}
Assumption $\cint=\cst$ implies that $\tilde{\rho}^*\in \cst$, which leads to $[U_t,\tilde{\rho}^*]=0$ for all $t$. Therefore,
\begin{equation}
   \rho^*_t=\tilde{\rho}^*,
\end{equation}
which proves the theorem.

\section{Proof of Lemma~\ref{lem:cint}}
\label{sec:proof_thm_cint}
(\emph{1}. $\to$ \emph{2}.) 
We first move to the interaction picture by \eqref{eq:transform_int}. Then, we prove that any initial state $\tilde{\rho}_0$ relaxes to some time-independent density matrix $\tilde{\rho}^*\in \cint$ which may depend on $\tilde{\rho}_0$. 

By defining 
\begin{equation}
    \aint=\langle \mathbb{I}, \{\tilde{L}_{m,t}\}\rangle,
\end{equation}
$\cint$ and $\aint$ are each other's commutant from von Neumann's bicommutant theorem, and the Hilbert space $\mathcal{H}$ can be decomposed into
representations of $\aint$ and $\cint$ as follows~\cite{moudgalya_HilbertSpaceFragmentation_2022}:
\begin{equation}
    \mathcal{H}=\bigoplus_{\alpha} \mathcal{H}^{(\mathcal{A})}_{\alpha} \otimes \mathcal{H}^{(\mathcal{C})}_{\alpha},
    \label{eq:hilbert_decomp}
\end{equation}
Here, $\mathcal{H}^{(\mathcal{A})}_{\alpha}$ and $\mathcal{H}^{(\mathcal{C})}_{\alpha}$ denotes 
$d_\alpha$ and $D_\alpha$ dimensional irreducible representation of $\aint$ and $\cint$. Under this decomposition, $O\in \cint$ and $\tilde{L}_{m,t}$ are written as
\begin{align} 
O=\bigoplus_{\alpha}
\mathbb{I}_{d_\alpha} \otimes O_\alpha,\ \tilde{L}_{m,t}=\bigoplus_{\alpha} 
\tilde{L}_{m,t,\alpha} \otimes \mathbb{I}_{D_\alpha}.
\label{eq:commutant_decomp}
\end{align}
According to~\eqref{eq:hilbert_decomp}, $\tilde{\rho}_t$ can be decomposed as 
\begin{equation}
\tilde{\rho}_t=\sum_{\alpha,\beta} \sum_k
\tilde{\rho}_{t,\alpha\beta,k}^{(\mathcal{A})} \otimes \tilde{\rho}_{t,\alpha\beta,k}^{(\mathcal{C})},
\label{eq:densitymatrix_decomp}
\end{equation}
where
\begin{align}
    \tilde{\rho}_{t,\alpha\beta,k}^{(\mathcal{A})}&: \mathcal{H}^{(\mathcal{A})}_{\beta}\to\mathcal{H}^{(\mathcal{A})}_{\alpha},\\
    \tilde{\rho}_{t,\alpha\beta,k}^{(\mathcal{C})}&: \mathcal{H}^{(\mathcal{C})}_{\beta}\to\mathcal{H}^{(\mathcal{C})}_{\alpha}.
\end{align}
Here, we used the Schmidt decomposition of the density matrix on $\mathcal{H}^{(\mathcal{A})}_{\alpha} \otimes \mathcal{H}^{(\mathcal{C})}_{\alpha}$ with the corresponding label $k$. By substituting \eqref{eq:commutant_decomp} and \eqref{eq:densitymatrix_decomp} to \eqref{eq:GKSL_int}, we have
\begin{align}
  \partial_t\tilde{\rho}_{t,\alpha\beta,k}^{(\mathcal{A})}
  &=-\frac{1}{2}\sum_{m}\mathcal{J}_{m,t,\alpha\beta}^2\left[\tilde{\rho}_{t,\alpha\beta,k}^{(\mathcal{A})}\right],\label{eq:GKSL_int_sector_A} 
  \\
  \partial_t\tilde{\rho}_{t,\alpha\beta,k}^{(\mathcal{C})}
  &=0.
  \label{eq:GKSL_int_sector_C}
\end{align}
where
\begin{equation}
    \mathcal{J}_{m,t,\alpha\beta}[A]:=\tilde{L}_{m,t,\alpha} A-A
  \tilde{L}_{m,t,\beta}.
\end{equation}
Since we consider the irreducible representations, Schur's lemma implies that
\begin{align}
   \cint_{\alpha \beta}
   &:=\bigcap_{m}\bigcap_{t=0}^\infty \Ker \mathcal{J}_{m,t,\alpha\beta}\\
   &=\begin{cases}
       \{c\mathbb{I}_{d_\alpha}\mid c\in\mathbb{C}\}& \text{if } \alpha=\beta, \\
       \emptyset& \text{if } \alpha\neq\beta.
   \end{cases} 
   \label{eq:cint_sector}
\end{align}
For $\alpha=\beta$, $\tilde{\rho}_{t,\alpha\alpha,k}^{(\mathcal{A})}\in \mathcal{B}(\mathcal{H}^{(\mathcal{A})}_{\alpha})$, and \eqref{eq:GKSL_int_sector_A} gives the GKSL equation in the interaction picture, thus Theorem~\ref{thm:main1} applies. Thus,
\begin{equation}    \lim_{t\to\infty}\tilde{\rho}_{t,\alpha\alpha,k}^{(\mathcal{A})}\propto\mathbb{I}
\end{equation}
for any initial condition from \eqref{eq:cint_sector}. 

For $\alpha\neq\beta$, $\tilde{\rho}_{t,\alpha\beta,k}^{(\mathcal{A})}$ is a mapping from a Hilbert space $\mathcal{H}^{(\mathcal{A})}_{\beta}$ to another Hilbert space $\mathcal{H}^{(\mathcal{A})}_{\alpha}$, thus Eq.~\eqref{eq:GKSL_int_sector_A} is not in the GKSL form. Nevertheless, by making slight modifications to the proof of Theorem~\ref{thm:main1}, one finds
\begin{equation}
 \lim_{t\to\infty}\tilde{\rho}_{t,\alpha\beta,k}^{(\mathcal{A})}=0.
 \label{eq:GKSL_off_diagonal}
\end{equation}
for any initial condition. Here, we make several remarks on Eq.~\eqref{eq:GKSL_off_diagonal}. First, for $\tilde{\rho}_{t,\alpha\beta,k}^{(\mathcal{A})}: \mathcal{H}^{(\mathcal{A})}_{\beta}\to\mathcal{H}^{(\mathcal{A})}_{\alpha}$ $(\alpha\neq\beta)$, the Hilbert-Schmidt norm is defined as 
\begin{equation}
\|\tilde{\rho}_{t,\alpha\beta,k}^{(\mathcal{A})}\|_2=\sqrt{\Tr[\tilde{\rho}_{t,\alpha\beta,k}^{(\mathcal{A})\dagger}\tilde{\rho}_{t,\alpha\beta,k}^{(\mathcal{A})}]}.
\end{equation}

Hereafter, we omit $\alpha,\beta,k$ and $(\mathcal{A})$. Then, by time evolution under \eqref{eq:GKSL_int_sector_A}, it does not increase, because
\begin{equation}
    \partial_t \|\tilde{\rho}_{t}\|_2=-\frac{\sum_{m=1}^M \Tr[(\mathcal{J}_{m,t,\alpha\beta}[\tilde{\rho}_{t}])^\dagger\mathcal{J}_{m,t,\alpha\beta}[\tilde{\rho}_{t}]]}{ \|\tilde{\rho}_{t}\|_2}\leq0.
\end{equation}
In particular, the equality holds if and only if $\mathcal{J}_{m,t,\alpha\beta}[\tilde{\rho}_{t}]=0$. From this and \eqref{eq:cint_sector}, by an argument analogous to that in the proof of Theorem~\ref{thm:main1}, one finds Eq.~\eqref{eq:GKSL_off_diagonal}.

Finally, from~\eqref{eq:GKSL_int_sector_C}, 
\begin{equation}
\tilde{\rho}_{t,\alpha\beta,k}^{(\mathcal{C})}=\tilde{\rho}_{0,\alpha\beta,k}^{(\mathcal{C})}
\end{equation}
for every $t$. Therefore, by writing $ \tilde{\rho}_{0,\alpha\beta}^{(\mathcal{C})}=\sum_{k}\tilde{\rho}_{0,\alpha\beta,k}^{(\mathcal{C})}$, we have
\begin{equation}
    \lim_{t\to\infty}\tilde{\rho}_t\propto\sum_{\alpha} \mathbb{I}_{d_\alpha} \otimes \tilde{\rho}_{0,\alpha\alpha}^{(\mathcal{C})}\in \cint.
\end{equation}

(\emph{2}. $\to$ \emph{1}.) If $\tilde{\rho}^* \in \cint$, then $ \tilde{\mathcal{L}}_t (\tilde{\rho}^*) = 0$ for all $t$. Thus, by choosing $\tilde{\rho}_0=\tilde{\rho}^*$, $\lim_{t\to \infty} \tilde{\rho}_t = \tilde{\rho}^*$. Thus, in the Schr\"odinger picture, $\lim_{t\to \infty} \rho_t = U_t\tilde{\rho}^*U_t^\dagger=\rho^*_t$.

\section{Proof of $\cst=\langle\mathbb{I},N\rangle$ in Examples \ref{ex:chinzei} and \ref{ex:hubbard_quasi}}
\label{sec:proof_hubbard_commutant}
In Examples \ref{ex:chinzei} and \ref{ex:hubbard_quasi}, we have seen that $\cst\supseteq\langle\mathbb{I},N\rangle$. In this section, we prove that the equality holds under appropriate assumptions on parameters. Since Example \ref{ex:chinzei} can be seen as a special case of Example \ref{ex:hubbard_quasi} ($B_1=B$, $\omega_1=\omega$, $B_2=0$), we focus on Example \ref{ex:hubbard_quasi}, whose Hamiltonian and jump operators are given by
\begin{align}
    H_t=
    &-\tau\sum_{\langle j,j^\prime\rangle\in\mathcal{N}}\sum_{\sigma=\uparrow,\downarrow}\left(c_{j,\sigma}^\dagger c_{j^\prime,\sigma}+\text{h.c.}\right)\nonumber\\
    &+\sum_{j\in\Lambda}\left[Un_{j,\uparrow}n_{j,\downarrow}+\mu_j n_j\right]\nonumber\\
    &+\frac{B_1}{2}\left[e^{-i\omega_1 t}S^++e^{i\omega_1 t}S^-\right]\nonumber\\
    &+\frac{B_2}{2}\left[e^{-i\omega_2 t}S^++e^{i\omega_2 t}S^-\right],\\
    L_j=&\sqrt{\kappa_j}n_j.
\end{align}
Here, $\Lambda$ denotes the set of sites and $\mathcal{N}$ denotes the set of bonds.
We assume that the lattice is connected, i.e., for any $j,k \in \Lambda$ such that $j\neq k$, there exists a finite sequence $l_1,\ldots,l_n \in \Lambda$ such that $l_1=j$, $l_n=k$, and $(l_i,l_{i+1})\in \mathcal{N}$ for $i = 1,\ldots,n-1$. 

The definition of $\cst$ is
\begin{align}
    \cst
    =\{O\in\mathcal{B}(\mathcal{H})\mid [H_t,O]=[L_{j},O]=0\ \forall j,t\}.
\end{align}
From von Neumann's bicommutant theorem~\cite{moudgalya_HilbertSpaceFragmentation_2022,moudgalya_SymmetriesCommutantAlgebras_2023},
\begin{equation}
    \cst=\langle\mathbb{I},N\rangle\iff \langle\mathbb{I}, \{H_t\},\{L_{j}\}\rangle=\langle\mathbb{I},N\rangle^\prime
\end{equation}
where $\langle\mathbb{I},N\rangle^\prime$ denotes the commutant of $\langle\mathbb{I},N\rangle$, i.e.,
   \begin{equation}
       \langle\mathbb{I},N\rangle^\prime=  \{O\in\mathcal{B}(\mathcal{H})\mid [O,N]=0\}.
   \end{equation} 
In our setup, the full operator algebra $\mathcal{B}(\mathcal{H})$ is spanned by polynomials in the creation and annihilation operators. Among these, the operators that commute with $N$ are spanned by polynomials containing equal numbers of creation and annihilation operators. Such polynomials are generated by $\mathbb{I}$ and $\{c^\dagger_{j,\sigma}c_{k,\tau} \mid   \sigma,\tau\in\{ \uparrow, \downarrow\},j,k\in\Lambda\}$. Therefore, it suffices to prove that
\begin{equation}
\langle\mathbb{I}, \{H_t\},\{L_{j}\}\rangle =\langle \mathbb{I},\{c^\dagger_{j,\sigma}c_{k,\tau}\}_{\sigma,\tau\in\{ \uparrow, \downarrow\},j,k\in\Lambda}\rangle.
\end{equation}
Since $H_t$ and $L_{j}$ can be written as a polynomial of $\{c^\dagger_{j,\sigma}c_{k,\tau}\}$, $\langle\mathbb{I}, \{H_t\},\{L_{j}\}\rangle \subseteq\langle \mathbb{I},\{c^\dagger_{j,\sigma}c_{k,\tau}\}_{\sigma,\tau\in\{ \uparrow, \downarrow\},j,k\in\Lambda}\rangle$ holds. Thus, it suffices to prove that
\begin{equation}
\langle\mathbb{I}, \{H_t\},\{L_{j}\}\rangle \supseteq\langle \mathbb{I},\{c^\dagger_{j,\sigma}c_{k,\tau}\}_{\sigma,\tau\in\{ \uparrow, \downarrow\},j,k\in\Lambda}\rangle.
\label{eq:hubbard_algebra_0}
\end{equation}
To this end, we prove the following:
\begin{align}
 S^+,S^-,S^z &\in \langle\mathbb{I}, \{H_t\},\{L_{j}\} \rangle,
\label{eq:hubbard_algebra_1}\\
\sum_{\sigma=\uparrow,\downarrow}c^\dagger_{j,\sigma}c_{k,\sigma} &\in \langle\mathbb{I}, \{H_t\},\{L_{j}\} \rangle\quad\forall j,k\in\Lambda,
\label{eq:hubbard_algebra_2}
\end{align}  

First, we prove~\eqref{eq:hubbard_algebra_1}. Noting that 
\begin{align}
    H_{t_i}-H_0
    = a_i S^++a_i^*S^-
    \in \langle\mathbb{I}, \{H_t\},\{L_{j}\}\rangle,
\end{align}
where
\begin{equation}
    a_i=\sum_{k=1,2}\frac{B_k}{2}(e^{-i\omega_k t_i}-1)
\end{equation}
and $t_i>0$, by choosing $t_1,t_2$ such that $a_1 a_2^*-a_1^*a_2\neq0$, we have
\begin{align}
    S^+,S^-
    &=\frac{a_2(H_{t_1}-H_0)-a_1(H_{t_2}-H_0)}{a_1 a_2^*-a_1^*a_2}\nonumber\\
    &\in \langle\mathbb{I}, \{H_t\},\{L_{j}\}\rangle,
\end{align}
which together with
\begin{equation}
    S^z=\frac{1}{2}[S^+,S^-]\in\langle\mathbb{I}, \{H_t\},\{L_{j}\}\rangle,
\end{equation}
implies \eqref{eq:hubbard_algebra_1}.

Next, we prove~\eqref{eq:hubbard_algebra_2}. When $\kappa_j>0$ for all $j$,
\begin{equation}
    n_j\in \langle\mathbb{I}, \{H_t\},\{L_{j}\}\rangle.
    \label{eq:hubbard_algebra_7}
\end{equation}
From~\eqref{eq:hubbard_algebra_1},  we have
\begin{align}
    T:=&H_0-\sum_{k=1,2}\frac{B_k}{2}(S^++S^-)\nonumber\\
     =&-\tau\sum_{\langle j,j^\prime\rangle}\sum_{\sigma=\uparrow,\downarrow}\left(c_{j,\sigma}^\dagger c_{j^\prime,\sigma}+\text{h.c.}\right)\nonumber\\     
     &+\sum_{j}\left[Un_{j,\uparrow}n_{j,\downarrow}+\mu_j n_j\right]\nonumber\\
     &\in\langle\mathbb{I}, \{H_t\},\{L_{j}\}\rangle.
\end{align}
Thus, by defining
\begin{equation}
    \mathcal{N}_j=\{k\in\Lambda\mid \langle j,k\rangle\in\mathcal{N}\},
\end{equation}
one finds
\begin{align}
[n_j,T]= -\tau&\sum_{l\in \mathcal{N}_j}\sum_{\sigma=\uparrow,\downarrow}(c^\dagger_{j,\sigma} c_{l,\sigma}-c^\dagger_{l,\sigma} c_{j,\sigma}), \\
[n_{k},[n_j,T]]&= \tau\sum_{\sigma=\uparrow,\downarrow}(c^\dagger_{j,\sigma} c_{k,\sigma}+c^\dagger_{k,\sigma} c_{j,\sigma}), \\
[n_{j},[n_{k},[n_j,T]]&= \tau\sum_{\sigma=\uparrow,\downarrow}(c^\dagger_{j,\sigma} c_{k,\sigma}-c^\dagger_{k,\sigma} c_{j,\sigma}).
\end{align}
Thus, if $\langle j,k\rangle\in\mathcal{N}$,
\begin{align}
\sum_{\sigma=\uparrow,\downarrow}c^\dagger_{j,\sigma} c_{k,\sigma}
&=\frac{[n_{k},[n_j,T]]+[n_{j},[n_{k},[n_j,T]]}{2\tau}\nonumber \\
&\in \langle\mathbb{I}, \{H_t\},\{L_{j}\}\rangle.
\end{align}
Moreover, since the lattice is connected, we can generate $\sum_{\sigma=\uparrow,\downarrow}c^\dagger_{l,\sigma}c_{m,\sigma}$ for any $l\neq m$ from $\sum_{\sigma=\uparrow,\downarrow}c^\dagger_{j,\sigma}c_{k,\sigma}$ such that $\langle j,k\rangle\in\mathcal{N}$ by successively using the commutation relation
\begin{equation}
\left[\sum_{\sigma=\uparrow,\downarrow}c^\dagger_{l,\sigma} c_{m,\sigma},\sum_{\tau=\uparrow,\downarrow}c^\dagger_{m,\tau} c_{n,\tau}\right]=\sum_{\sigma=\uparrow,\downarrow}c^\dagger_{l,\sigma} c_{n,\sigma}
\end{equation}
for $l\neq n$. Thus,
\begin{equation}
\sum_{\sigma=\uparrow,\downarrow}c^\dagger_{l,\sigma}c_{m,\sigma}\in\langle\mathbb{I}, \{H_t\},\{L_{j}\}\rangle.
\label{eq:hubbard_algebra_8}
\end{equation}
From \eqref{eq:hubbard_algebra_7} and \eqref{eq:hubbard_algebra_8}, we have \eqref{eq:hubbard_algebra_2}.

Next, from \eqref{eq:hubbard_algebra_1} and \eqref{eq:hubbard_algebra_2}, we prove \eqref{eq:hubbard_algebra_0}. First, from $N,S^z\in \langle\mathbb{I}, \{H_t\},\{L_{j}\}\rangle$, we generate the projection operator $P_{\uparrow(\downarrow)}$ to subspace of $\mathcal{H}$ with fermions with only up (down) spins. Since
\begin{equation}
    \prod_{\substack{m=0 \\ (m\neq n)}}^{2|\Lambda|}\frac{N-m\mathbb{I}}{n-m} \in \langle\mathbb{I}, \{H_t\},\{L_{j}\}\rangle
\end{equation}
is a projection operator to the eigenspace $N=n$ and 
\begin{equation}
    \prod_{l=-\frac{n}{2}}^{\frac{n}{2}-1}\frac{S^z-l\mathbb{I}}{\frac{n}{2}-l},\quad\prod_{l=-\frac{n}{2}+1}^{\frac{n}{2}}\frac{S^z-l\mathbb{I}}{-\frac{n}{2}-l}\in \langle\mathbb{I}, \{H_t\},\{L_{j}\}\rangle
\end{equation}
is a projection operator to the eigenspace $S^z=\frac{n}{2}$ and $S^z=-\frac{n}{2}$, $P_{\uparrow(\downarrow)}$ can be written as 
\begin{align}
    P_{\uparrow}&=\sum_{n=0}^{2|\Lambda|} \left[\prod_{\substack{m=0 \\ (m\neq n)}}^{2|\Lambda|}\frac{N-m\mathbb{I}}{n-m} \prod_{l=-\frac{n}{2}}^{\frac{n}{2}-1}\frac{S^z-l\mathbb{I}}{\frac{n}{2}-l}\right]\nonumber\\
    &\in \langle\mathbb{I}, \{H_t\},\{L_{j}\}\rangle,\\
    P_{\downarrow}&=\sum_{n=0}^{2|\Lambda|} \left[\prod_{\substack{m=0 \\ (m\neq n)}}^{2|\Lambda|}\frac{N-m\mathbb{I}}{n-m} \prod_{l=-\frac{n}{2}+1}^{\frac{n}{2}}\frac{S^z-l\mathbb{I}}{-\frac{n}{2}-l}\right]\nonumber\\
    &\in \langle\mathbb{I}, \{H_t\},\{L_{j}\}\rangle.
\end{align}
Therefore,
\begin{align}
    c^\dagger_{j,\uparrow}c_{k,\uparrow} =P_{\uparrow} \sum_{\sigma=\uparrow,\downarrow}c^\dagger_{j,\sigma}c_{k,\sigma}P_{\uparrow}\in \langle\mathbb{I}, \{H_t\},\{L_{j}\}\rangle,
    \label{eq:hubbard_algebra_3}\\
    c^\dagger_{j,\downarrow}c_{k,\downarrow} =P_{\downarrow} \sum_{\sigma=\uparrow,\downarrow}c^\dagger_{j,\sigma}c_{k,\sigma}P_{\downarrow}\in \langle\mathbb{I}, \{H_t\},\{L_{j}\}\rangle.
    \label{eq:hubbard_algebra_4}
\end{align}
Finally, we have
\begin{align}
   c^\dagger_{j,\uparrow}c_{k,\downarrow}&=[c^\dagger_{j,\uparrow}c_{k,\uparrow},S^+]\in \langle\mathbb{I}, \{H_t\},\{L_{j}\}\rangle,
   \label{eq:hubbard_algebra_5}\\
   c^\dagger_{j,\downarrow}c_{k,\uparrow}&=-[c^\dagger_{j,\uparrow}c_{k,\uparrow},S^-]\in \langle\mathbb{I}, \{H_t\},\{L_{j}\}\rangle.
   \label{eq:hubbard_algebra_6}
\end{align}
From~\eqref{eq:hubbard_algebra_3}, \eqref{eq:hubbard_algebra_4}, \eqref{eq:hubbard_algebra_5}, and \eqref{eq:hubbard_algebra_6}, we have \eqref{eq:hubbard_algebra_0}, which completes the proof.

\bibliography{reference}

\end{document}